%% file: main.tex
\documentclass[conference]{IEEEtran}
\usepackage{amsmath, amsfonts, amsthm}
\usepackage{romannum}
\usepackage{graphicx}
\usepackage[ruled,linesnumbered]{algorithm2e}
\SetAlgoSkip{3pt}
\usepackage{subcaption}
\usepackage{bm}
\usepackage{xcolor}
\usepackage{multirow}
\usepackage{enumitem}
\usepackage{xspace}
\usepackage{stfloats}
\usepackage{flushend}
\usepackage{hyperref}
\usepackage{cuted}
\usepackage{url}

\newtheorem{definition}{Definition}
\newtheorem{theorem}{Theorem}

\newcommand{\RIAR}{\ensuremath{\mathsf{RIAR}}\xspace}

\author{
    \IEEEauthorblockN{
        Xiaoguang~Li\IEEEauthorrefmark{1},
        Hanyi~Wang\IEEEauthorrefmark{2},
        Yaowei~Huang\IEEEauthorrefmark{3},
        Jungang~Yang\IEEEauthorrefmark{4},
        Qingqing~Ye\IEEEauthorrefmark{5},
        Haonan~Yan\IEEEauthorrefmark{1},
        Ke~Pan\IEEEauthorrefmark{1}, \\
        Zhe~Sun\IEEEauthorrefmark{3},
        Hui~Li\IEEEauthorrefmark{1}
        }
        \IEEEauthorblockA{\IEEEauthorrefmark{1}Xidian University, \IEEEauthorrefmark{2}China Mobile (Suzhou) Software Technology Co., Ltd., \IEEEauthorrefmark{3}Guangzhou University}
        \IEEEauthorblockA{\IEEEauthorrefmark{4}Shanghai University, \IEEEauthorrefmark{5}The Hong Kong Polytechnic University}
}

\begin{document}
\title{Robust Single-message Shuffle Differential Privacy Protocol for Accurate Distribution Estimation}







\maketitle
\begin{abstract}
    Shuffler-based differential privacy (shuffle-DP) is a privacy paradigm providing high utility by involving a shuffler to permute noisy report from users. Existing shuffle-DP protocols mainly focus on the design of shuffler-based categorical frequency oracle (SCFO) for frequency estimation on categorical data. However, numerical data is a more prevalent type and many real-world applications depend on the estimation of data distribution with ordinal nature. In this paper, we study the distribution estimation under pure shuffle model, which is a prevalent shuffle-DP framework without strong security assumptions. We initially attempt to transplant existing SCFOs and the na\"ive distribution recovery technique to this task, and demonstrate that these baseline protocols cannot simultaneously achieve outstanding performance in three metrics: 1) utility, 2) message complexity; and 3) robustness to data poisoning attacks. Therefore, we further propose a novel single-message \textit{adaptive shuffler-based piecewise} (ASP) protocol with high utility and robustness. In ASP, we first develop a randomizer by parameter optimization using our proposed tighter bound of mutual information. We also design an \textit{Expectation Maximization with Adaptive Smoothing} (EMAS) algorithm to accurately recover distribution with enhanced robustness. To quantify robustness, we propose a new evaluation framework to examine robustness under different attack targets, enabling us to comprehensively understand the protocol resilience under various adversarial scenarios. Extensive experiments demonstrate that ASP outperforms baseline protocols in all three metrics. Especially under small $\epsilon$ values, ASP achieves an order of magnitude improvement in utility with minimal message complexity, and exhibits over threefold robustness compared to baseline methods.
\end{abstract}

\pagestyle{plain}
\pagenumbering{arabic}



\input{intro_v8}
\input{preliminary_v3}
\input{baseline}
\input{proposed_method_v4}

\input{utility_evaluation_v2}
\input{robustness_evaluation_framework}
\input{robustness_evaluation_v2}

\input{related_work}

\bibliography{ref}
\bibliographystyle{IEEEtran}

\end{document}

%% file: intro_v8.tex
\section{Introduction}

Differential Privacy (DP)~\cite{dwork2006differential, dwork2008differential} is the \textit{de facto} standard for privacy-preserving data analysis.
In the central model, a trusted server collects users' raw data and publishes the analysis results perturbed by noise with the scale of $\Theta(1)$.
Local differential privacy (LDP)~\cite{duchi2013local, gu2020pckv,li2020estimating,wang2017locally,wang2019locally,wang2020locally,wang2021continuous, apple, erlingsson2014rappor, ding2017collecting} removes this trust assumption by locally randomizing data, but sacrifices utility -- noise grows to $\Theta(\sqrt{n})$ for $n$ users.
To achieve the middle ground between DP and LDP, shuffle-DP~\cite{erlingsson2019amplification, balle2019privacy, cheu2019distributed, murakami2025augmented} involves an intermediate shuffler between the local randomizer and the server to anonymize users by shuffling their reports, achieving better privacy and utility.
Within shuffle-DP, the \textit{pure shuffle model}~\cite{erlingsson2019amplification, wang13improving, balle2020private, ghazi2021differentially, dong2024almost, balcer2021connecting, cheu2022differentially, li2023privacy} only requires the shuffler to perform shuffling, whereas the \textit{augmented shuffle model}~\cite{murakami2025augmented} allows the shuffler to perform additional operations such as sampling and noise generation. These extra computations require the shuffler to be trustworthy and entails more security considerations in deployment. If any of these computations is compromised, the privacy is completely violated and the accuracy is diminished. Therefore, we focus on the pure shuffle model, which relies on weaker assumptions and is less susceptible to such hazards.

Existing protocols in pure shuffle model usually focus on categorical data and design shuffler-based categorical frequency oracle (SCFO)~\cite{wang13improving, balle2020private, cheu2022differentially, li2023privacy} for frequency estimation.
However, few work explored distribution estimation on numerical data with ordinal nature, which is a more prevalent data type and widely-encountered in practice.
As a motivated example, consider a government agency that wants to analyze the income distribution of citizens to design tax brackets and social benefits, but citizens do not trust the agency with their exact income. In this dilemma, the agency can aggregate users' incomes under shuffle-DP and make financial policy accordingly.

\noindent\textbf{Baseline methods.} We first propose two baseline approaches by directly expanding existing techniques to numerical domain. The \textbf{first} is to bin the numerical domain and apply SCFOs with consistency calibration, and the \textbf{second} is to simply shuffle the outcome of state-of-the-art LDP protocol for distribution estimation (e.g., SW~\cite{li2020estimating}).
Typically, shuffle-DP protocols are evaluated by 1) utility (estimate accuracy) and 2) message complexity (number of messages sent per user).
Besides, recent studies~\cite{li2024robustness, cheu2022differentially, murakami2025augmented} revealed that shuffle-DP protocols are vulnerable to data poisoning attacks where an attacker can send bogus data to the server and manipulate the final estimate.
Therefore, robustness (resilience to poisoning attack) becomes the third important metric in shuffle-DP.

We demonstrate that the baseline protocols cannot simultaneously achieve satisfactory results across all three metrics.
(1) Existing works fail to offer favorable utility. SCFOs treat numerical domain as multiple discrete chunks and ignore the meaningful order nature. The parameter in LDP protocol (e.g., SW~\cite{li2020estimating}) is not fine-tuned for the shuffle framework. Existing aggregations typically use fixed coefficients to recover the distribution, limiting the satisfactory performance only on dataset with specific characteristics; 
(2) Some SCFOs, e.g., Flip~\cite{cheu2022differentially} and Pure~\cite{li2023privacy}, rely on multi-message reporting mechanism to reduce noise, which however suffers from high message complexity;
(3) Our experiments (see Section~\ref{sec:robustness_evaluation}) show that multi-message reporting enlarges the attack surface since more bogus data can be sent. Due to the lack of robust aggregation, the resilience of SCFOs is further deteriorated.

\noindent\textbf{Our proposal.} This paper reports our experience of overcoming the above limitations and proposes a novel single-message shuffle-DP protocol called Adaptive Shuffler-based Piecewise (ASP) for accurate distribution estimation with high robustness. 
Two challenges arise in ASP design: 1) how to leverage the shuffling property to design a randomizer that can preserve the maximal data information (high utility) by the lowest message complexity and 2) how to construct an aggregator to consistently recover accurate distribution (high utility) while mitigating attack impact (high robustness).

To address the first challenge, we leverage the shuffling property to substantiate our single-message randomizer by two tunable parameters instead of a fixed privacy budget, which are optimized by our new tighter bound of the mutual information between the perturbed value and the true data.
For the second challenge, we propose a novel smoothing-based aggregation called Expectation-Maximization with Adaptive Smoothing (EMAS) by adopting the smoothing technique, which has been demonstrated to be a promising technique to improve the utility and robustness~\cite{li2024robustness}. Different from the prior work using fixed coefficients for smoothing, EMAS combines perturbation intensity (controlled by $\epsilon$), data information and weight decay~\cite{d2024we} for high-likelihood estimate and superior robustness.

\noindent\textbf{Robustness evaluation.}
While there exist well-established metrics for utility and message complexity such as Wasserstein distance and the number of messages, the robustness evaluation framework still needs to be improved.
Existing framework~\cite{li2024robustness} only considers a single attack setting in evaluation, impeding it to fully reveal the protocol robustness under various adversarial scenarios.
We propose a novel attack-driven robustness evaluation framework by considering a more general and flexible attacker aiming to maliciously shift densities to desired targets.
We also propose a new metric called Real and Ideal Attack Ratio (RIAR) to quantify the robustness with respect to the targets.
RIAR quantifies the efficacy between real and ideal attacks and higher RIAR indicates better robustness, which helps us holistically understand the protocol robustness under a larger attack surface.


We evaluate our ASP by three statistical tasks on one synthetic dataset and three real-world datasets.
The results demonstrate that our protocol has the minimum message complexity and outperforms other protocols in most cases.
Specifically, our ASP can reduce the estimate error by almost half compared to baseline protocols for all statistical tasks under small $\epsilon$ value (e.g., 0.01).
Moreover, when the data distribution is pathological (e.g., spiky and jagged), our protocol performs much better and shows an improvement of an order of magnitude over baseline protocols.
By our metric RIAR, the evaluation shows that our ASP is the most robust among all protocols. 
When $\epsilon \leq 0.04$ and 5\% users are compromised, SCFO-based protocols fail to resist the attack and the attack can achieve near-ideal performance.
However, ASP demonstrates a high attack robustness and it exhibits more than threefold RIAR compared to baseline protocols, which indicates attack efficacy on ASP significantly deviates from the ideal attack.
We summarize our contributions as follows.

\begin{itemize}[leftmargin=*]
    \item We propose a versatile single-message shuffle-DP protocol ASP for numerical distribution estimation that sufficiently leverages the ordinal property of the domain.

    \item We design a new adaptive-smoothing based aggregation EMAS to restore the distribution, which enhances the robustness while consistently achieving high utility across diverse dataset.

    \item We propose a robustness evaluation framework using a more general attack and a new target-dependent metric to measure the robustness arising from the protocol design.

    \item We empirically study the proposed protocol using several datasets. The results demonstrate that our protocol outperforms existing methods in terms of utility, message complexity and robustness.
\end{itemize}

%% file: preliminary_v3.tex
\vspace{-10pt}
\section{Preliminary}
\subsection{DP and LDP}
In DP, a trusted server collects users' private data, analyzes the data, and then publishes the perturbed result with DP noise.
Intuitively, DP requires that any single element in the dataset only has a limited impact on the output.
\vspace{-5pt}
\begin{definition}[$(\epsilon, \delta)$-Differential Privacy]
    An algorithm $M$ satisfies $(\epsilon, \delta)$-differential privacy if and only if for any neighbor datasets $D_1$ and $D_2$ that only differ on one element, and the output range $S$, the following inequality holds
    \begin{align*}
        \Pr[M(D_1) \in S] \leq e^\epsilon \Pr[M(D_2) \in S] + \delta.
    \end{align*}
\end{definition}

LDP considers the setting in which there is an untrusted server and $n$ users, each possessing a private value $x \in \mathcal{D}$. To protect privacy, each user perturbs his value $x$ by a randomizer $\mathcal{R}$ and reports $\mathcal{R}(x)$ to the server. The algorithm $\mathcal{R}$ is $(\epsilon, \delta)$-LDP if and only if it satisfies the following definition.
\vspace{-5pt}
\begin{definition}[$(\epsilon, \delta)$-Local Differential Privacy]
    An algorithm $\mathcal{R}$ satisfies $(\epsilon, \delta)$-LDP if and only if for any pair of inputs $x_1, x_2 \in \mathcal{D}$, the following inequality holds
    \begin{align*}
        \Pr[\mathcal{R}(x_1) = \hat{x}] \leq e^\epsilon \Pr[\mathcal{R}(x_2) \in \hat{x}] + \delta
    \end{align*}
\end{definition}

Given the users' LDP report, the server then analyzes the statistical result by an aggregation function $\mathcal{A}$. When $\delta = 0$, we simplify $(\epsilon, 0)$-DP and $(\epsilon, 0)$-LDP to be $\epsilon$-DP and $\epsilon$-LDP. Then we review the state-of-the-art LDP protocol for distribution estimation used in this paper.

\noindent\textbf{SW mechanism~\cite{li2020estimating}.}
SW assumes the input domain is $[0, 1]$ since any bounded value space can be linearly transformed into this domain.
It considers the ordinal information of the numerical domain and perturbs each input $x$ by the following square-wave reporting mechanism
\begin{align*}
    \Pr[\mathcal{R_{SSW}}(x)=\hat{x}] =
    \begin{cases}
        p, & \mathrm{if} |x - \hat{x}| \leq b \\
        q, & \mathrm{otherwise}
    \end{cases},
\end{align*}
where $p = \frac{e^\epsilon}{2be^\epsilon + 1}$, $q = \frac{1}{2be^\epsilon + 1}$ and $b = \frac{\epsilon e^\epsilon - e^\epsilon + 1}{2e^\epsilon (e^\epsilon - 1 - \epsilon)}$. Therefore, the output domain is $[-b, 1+b]$.
SW then uses the Expectation-Maximization with Smoothing (EMS) as the aggregation function $\mathcal{A}$ to recover a maximum-likelihood distribution.
EMS is a variant of Expectation-Maximization (EM) algorithm~\cite{silvey2017statistical}. It iteratively executes expectation (E-step), maximization (M-step) and smoothing (S-step) to finally output a histogram with $m$ bins $\bm{\tilde{f}} = [\tilde{f}_1, ..., \tilde{f}_m]$ as the estimated distribution, where $\tilde{f}_i$ is the estimated frequency of the $i$-th bin.
The E-step produces a function $Q(\tilde{\bm{f}} | \tilde{\bm{f}}^t)$ that evaluates the log-likelihood expectation of $\tilde{\bm{f}}$, given the estimate $\tilde{\bm{f}}^t = [\tilde{f}_1^t, ..., \tilde{f}_m^t]$ in the $t$-th iteration.
The M-step calculates the estimate $\tilde{\bm{f}}^* = \arg\max_{\tilde{\bm{f}}}  Q(\tilde{\bm{f}} | \tilde{\bm{f}}^t)$.
The S-step averages each estimate with its adjacent ones using binomial coefficients and has $\tilde{f}_{i}^{t+1} = \frac{1}{2}\tilde{f}_{i}^* + \frac{1}{4}(\tilde{f}_{i-1}^* + \tilde{f}_{i+1}^*)$.

\subsection{Pure Shuffle Model and Robustness}
\noindent\textbf{Pure Shuffle Model.}
A protocol in the pure shuffle model consists of three algorithms: 

\textsf{Randomizer} $\mathcal{R}: \mathcal{D} \rightarrow \hat{\mathcal{D}}^{w}$. The local randomizer $\mathcal{R}$ takes the $i$-th user's data $x^{(i)}$ as input and outputs a set of perturbed messages $\hat{x}^{(i)}_{1}, \hat{x}^{(i)}_{2}, ..., \hat{x}^{(i)}_{w} \in \hat{\mathcal{D}}^{w}$.

\textsf{Shuffler} $\mathcal{S}: \hat{\mathcal{D}}^{w} \rightarrow \hat{\mathcal{D}}^{w}$. Shuffler $\mathcal{S}$ applies a uniformly random permutation to the input messages.

\textsf{Aggregation} $\mathcal{A}: \hat{\mathcal{D}}^{w} \rightarrow \mathcal{Z}$. Aggregation takes the output of the shuffler and analyzes the statistical result.

We concentrate on the standard case where an independent secure $\mathcal{S}$~\cite{bogdanov2008sharemind, laur2011round, wang13improving} is deployed to avoid the collusion between $\mathcal{A}$ and $\mathcal{S}$. This case is also consistent with the prior work~\cite{cheu2022differentially} on SCFO design and robustness study.
A shuffle-DP protocol is $(\epsilon, \delta)$-DP if and only if the output of the shuffled result $\mathcal{S} \circ \mathcal{R}$ satisfies $(\epsilon, \delta)$-DP~\cite{cheu2022differentially}.

The utility of shuffle-DP protocols is measured by statistical similarity between the estimate and the ground truth, and higher similarity means better utility. In this paper, we consider three popular statistical tasks for distribution utility measurement: 1) range query, 2) quantile and 3) Wasserstein distance.
The message complexity is the number $w$ of messages sent per user and lower $w$ means better efficiency. A shuffle-DP protocol is called \textit{single-message} protocol if $w = 1$ or \textit{multi-message} protocol if $w > 1$.

\noindent\textbf{Robustness.}
Due to the distributed nature, LDP and shuffle-DP protocols are inherently vulnerable to data poisoning attacks~\cite{cheu2021manipulation, cao2021data, li2023fine,murakami2025augmented, cheu2022differentially, balle2020private}, in which an attacker compromises a fraction $\beta (0 \leq \beta \leq 1)$ of $n$ users to send forged bogus data $\hat{Y}$ to malicious manipulate the final estimate.
We expand the prior robustness evaluation framework by deploying a more flexible attacker and a new target-dependent robustness metric (see Section~\ref{sec:robustness_evaluation_framework} for details).

%% file: baseline.tex
\section{Baseline Method} \label{sec:baseline_method}

In this section, we propose two types of baseline shuffle-DP protocols to reconstruct the distribution. The first is SCFO with binning and consistency and the second is the shuffler-based SW (SSW), which is the best-effort adaption of the state-of-the-art LDP protocol for shuffle-DP.

\subsection{SCFO with Binning and Consistency}
There are several SCFOs~\cite{luo2022frequency, balcer2019separating, ghazi2020private, cheu2022differentially, li2023privacy, wang13improving} in pure shuffle model. However, some protocols~\cite{luo2022frequency, balcer2019separating, wang13improving} have poor utility and only work on datasets of specific size, and protocol~\cite{ghazi2020private} is even vulnerable to floating-point attacks~\cite{mironov2012significance}.
Therefore, we choose the state-of-the-art SCFOs Flip~\cite{cheu2022differentially} and Pure~\cite{li2023privacy} as the building blocks.
We first revisit the SCFO for frequency estimation of any value $x \in \mathcal{D}$, and we assume $m = |\mathcal{D}|$ is the domain size and use the notation $[m]$ to denote the set $\{1, 2, ..., m\}$ for convenience.


\noindent\textbf{Flip protocol.}
Flip builds upon artificial dummy users who involve dummy messages in the local side for better utility.
Concretely, each user $i$ first encodes his data $x^{(i)}$ as a length-$m$ one-hot vector $\bm{x}^{(i)}$ where all elements are zeros except that the element at position $x^{(i)}$ is 1.
Then he generates $s$ length-$m$ all-zero vectors as $s$ dummy users.
Each user flips all elements in $s+1$ vectors with probability $q$.
After shuffling, $ns+n$ messages of $n$ users satisfy $(\epsilon, \delta)$-DP, where $s = \max(\frac{132}{5n}(\frac{e^\epsilon + 1}{e^\epsilon - 1})^2 \ln{\frac{4}{\delta}}, \frac{2}{n} \ln{20m} - 1)$ and $q(1-q) = \frac{33}{5ns}(\frac{e^\epsilon + 1}{e^\epsilon - 1})^2 \ln{\frac{4}{\delta}}$.
To estimate the frequency $\tilde{f}_j$ of the bin $B_j$ corresponding to $j \in [m]$, the server conducts the aggregation $\mathcal{A}(B_j) = \frac{1}{n} \sum_{i=1}^{n} \sum_{v=1}^{s+1} \frac{1}{1-2q} (\hat{\bm{x}}^{(i)}_{v}[j] - q)$, where $\hat{\bm{x}}^{(i)}_{v}[j]$ is the element on position $j$ in the $v$-th vector of the $i$-th user.

\vspace{3pt}

\noindent\textbf{Pure protocol.}
Pure provides privacy under the umbrella of a \textit{dummy blanket} technique.
Pure first generates $s$ random dummy points $\hat{x}^{(i)}_1, ..., \hat{x}_{s}^{(i)}$ from the discrete domain $[1, 2, ..., m]$ for the $i$-th user, who then mixes the dummy points with his private data $x^{(i)}$ and sends the result $\hat{\bm{x}}^{(i)} = [x^{(i)}, x^{(i)}_1, ..., x^{(i)}_s]$ to $\mathcal{S}$.
The output of the shuffler satisfies $(\epsilon, \delta)$-DP where $\epsilon = \sqrt{\frac{14m\cdot \ln{\frac{2}{\delta}}}{ns-1}}$.
The frequency $\tilde{f}_j$ of the bin $B_j$ corresponding to $j \in [m]$ is estimated by the aggregation $\mathcal{A}(B_j) = \frac{1}{n} (\sum_{i=1}^{n} \sum_{v=1}^{s+1} \mathbb{I}_{[j]}(\hat{\bm{x}}^{(i)}_v) - \frac{ns}{m})$, where $\hat{\bm{x}}^{(i)}_v$ is the $v$-th value of the $i$-th user and $\mathbb{I}_{[j]}(\hat{\bm{x}}^{(i)}_v)$ is the indicator function which returns 1 if $\hat{\bm{x}}^{(i)}_k = j$ and 0 otherwise.

Then we introduce how to adapt SCFOs for distribution estimation.
We assume the numerical input domain is $[0, 1]$ since any bounded value can be mapped into this domain.
To accommodate the goal of SCFO, each user first discretizes the numerical domain into $m$ bins and reports the index $x_b$ of the bin containing the private data $x$ using a SCFO with $\epsilon$ and $\delta$.
Then, the server applies the Norm-Sub as a post-processing consistency step to calibrate the estimate without extra privacy cost, yielding a valid estimated distribution (all frequencies are non-negative and sum-to-1).

\subsection{SSW Mechanism}
SSW directly applies the shuffler after the perturbation in SW, and the server also conducts the EMS to recover the distribution.
We present the privacy guarantee of SSW below.

\begin{theorem}\label{the:privacy_ssw}
    If the SW mechanism is built upon parameters $p = \frac{e^{\epsilon_l}}{2be^{\epsilon_l} + 1}$, $q = \frac{1}{2be^{\epsilon_l} + 1}$, $b = \frac{\epsilon_l e^{\epsilon_l} - e^{\epsilon_l} + 1}{2e^{\epsilon_l} (e^{\epsilon_l} - 1 - \epsilon_l)}$ and achieves $\epsilon_l$-LDP, the SSW is $(\epsilon, \delta)$-DP for any $\epsilon, \delta$ such that
    \begin{align*}
        \frac{r^2}{4 (1+2b)q n (e^\epsilon - 1)} e^{-(1+2b)q n (1 - e^{-\frac{2(e^\epsilon - 1)^2}{r^2}})} \leq \delta,
    \end{align*}
    where $r = (1 + e^\epsilon)(\frac{p-q}{1+2b})$.
\end{theorem}



\begin{proof}
    The privacy is quantified with the help of privacy blanket decomposition~\cite{balle2019privacy}.
    Since SW mechanism is $\epsilon_l$-LDP, it can be expressed by a mixture decomposition $\mathcal{R}(x) = (1 - \gamma)v_x + \gamma \omega$, where $\gamma$ is a constant, $v_x$ is a random variable dependent of the input $x$ and $\omega$ is the random variable independent of $x$, called \textit{blanket}.
    We use $u_x(\hat{x})$ to denote the probability density that SSW mechanism outputs $\hat{x}$ given $x$, then we have
    \begin{align*}
        & \gamma = \int_{\hat{x}} \inf\limits_\mathrm{x} u_x(\hat{x}) dy = \int_{-b}^{1+b} q dy = (1+2b)q \\
        & \omega(\hat{x}) = \frac{\inf\limits_\mathrm{x} u_x(\hat{x})}{\gamma} = \frac{q}{(1+2b)q} = \frac{1}{1+2b},
    \end{align*}
    where $\omega(\hat{x})$ is the density of $\omega$ at $\hat{x}$.
    To measure the privacy of the output of the shuffler in SSW, we denote all users input as $X$ and have the privacy amplification random variable $L$ as follows $L = \frac{u_{x_1}(\hat{x}) - e^\epsilon u_{x_2}(\hat{x})}{\omega(\hat{x})}$.
    Given the density of $\omega(\hat{x})$, we have $\frac{q}{1 + 2b} \leq \frac{u_{x_1}(\hat{x})}{\omega(\hat{x})} \leq \frac{p}{1 + 2b}$.
    Then we can calculate the bound of $L$,
    \begin{align*}
        \frac{q}{1 + 2b} - e^\epsilon \frac{p}{1 + 2b} \leq L \leq \frac{p}{1 + 2b} - e^\epsilon \frac{q}{1 + 2b}
    \end{align*}
    
    There is an strong relationship between samples $L_1, L_2, ...$ of $L$ and hockey-stick divergence of the output of the shuffler $\mathfrak{D}_{e^\epsilon}(\mathcal{S} \circ \Psi(X) || \mathcal{S} \circ \Psi(X'))$, where $X$ and $X'$ are a pair of neighbor datasets.
    \begin{align*}
        &\quad \mathfrak{D}_{e^\epsilon}(\mathcal{S} \circ \Psi(X) || \mathcal{S} \circ \Psi(X')) \\
        &\leq \frac{1}{\gamma n} \sum_{m=1}^{n} {n \choose m} \gamma^m (1 - \gamma)^{n-m} \mathbb{E}\left(\sum_{i=1}^{m} L_i \right)_{+},
    \end{align*}
    where $(\mu)_{+} = \max{\{0, \mu\}}$. According to the characteristic of DP that a mechanism $M$ is $(\epsilon, \delta)$-DP if and only if $\mathfrak{D}_{e^\epsilon}(M(X) || M(X')) \leq \delta$, the SSW mechanism is $(\epsilon, \delta)$-DP if and only if
    \begin{align}
        \frac{1}{\gamma n} \sum_{m=1}^{n} {n \choose m} \gamma^m (1 - \gamma)^{n-m} \mathbb{E}\left(\sum_{i=1}^{m} L_i \right)_{+} \leq \delta
        \label{eq: dp_divergence}
    \end{align}

    Let $b_l \leq L_i \leq b_u$ and denote $r = b_u - b_l$, by the Hoeffding’s inequality, we have
    \begin{align}
        \mathbb{E}\left(\sum_{i=1}^{m} L_i \right)_{+} \leq \frac{r^2}{4a} e^{-\frac{2ma^2}{r^2}}
        \label{eq: upper_bound_L}
    \end{align}
    where $a = e^{\epsilon} - 1$. By substituting Equation~\ref{eq: upper_bound_L} into Equation~\ref{eq: dp_divergence}, we have
    \begin{align*}
        \mathrm{Equation~(1)} \leq \frac{r^2}{4 \gamma n a} e^{-\gamma n (1 - e^{-\frac{2a^2}{r^2}})}
    \end{align*}
    Therefore, the SSW mechanism satisfies $(\epsilon, \delta)$-DP if and only if the following inequality holds
    \begin{align*}
        \frac{r^2}{4 (1+2b)q n (e^\epsilon - 1)} e^{-(1+2b)q n (1 - e^{-\frac{2(e^\epsilon - 1)^2}{r^2}})} \leq \delta,
    \end{align*}
    where $r = (1 + e^\epsilon)(\frac{p-q}{1+2b})$.
\end{proof}

Theorem~\ref{the:privacy_ssw} demonstrates that the SSW can amplify the privacy and achieve more accurate estimate since $\epsilon < \epsilon_l$.


\subsection{Observation}
Both baseline methods attempt to estimate the distribution in pure shuffle model by expanding the existing techniques. However, we observe that they both have shortcomings.
The first technique SCFO ignores the nature of the numerical domain and produces multiple messages for better accuracy at the cost of communication overhead.
The second technique SSW uses fixed parameters to provide unnecessary $\epsilon_l$-LDP guarantee on the local side, limiting the search space of the optimal parameter. Consequently, the parameters of SSW are far away from the optimal. Our experiments validate this conclusion, as will be elaborated in Section~\ref{sec:utility_evaluation}.

%% file: proposed_method_v4.tex
\section{Proposed Method}

\subsection{Overview}
To overcome the shortcomings in the baseline methods, we propose the single-message shuffle-DP protocol called ASP which achieves better accuracy and robustness. It consists of a new randomizer $\mathcal{R}_{ASP}$ on the local side and an adaptive smoothing-based aggregator called EMAS on the server.
The main workflow is shown in Figure~\ref{fig:asp_workflow}. Each user first inputs the private data $x^{(i)}$ to $\mathcal{R}_{ASP}$ and sends the result $\hat{x}^{(i)} = \mathcal{R}_{ASP}(x^{(i)})$ to the shuffler $\mathcal{S}$, which then permutes all noisy reports $[\hat{x}^{(1)}, ..., \hat{x}^{(n)}]$. Finally, the server receives the shuffled results and conducts the EMAS to recover the data distribution.

We have two new designs in ASP. \textbf{First}, $\mathcal{R}_{ASP}$ inherits the idea of squared-wave reporting but adopts two tunable parameters to achieve a finer perturbation control than SSW/SW. Then, we propose a tighter upper bound of the mutual information between the input and the output of $\mathcal{R}_{ASP}$. Using this tighter bound together with the privacy amplification of shuffle-DP, we optimize the perturbation parameters for better utility. However, SW/SSW uses a loose upper bound of mutual information for parameter optimization, resulting in sub-optimal parameters and worse utility. For example, given $n=100000$, $\epsilon = 0.01$ and $\delta=10^{-5}$, $\mathcal{R}_{ASP}$ has $b=0.215$, $p = 1.13$ while $\mathcal{R}_{SSW}$ has $b=0.332$, $p = 0.827$, which means the output $\hat{x} = \mathcal{R}_{ASP}(x)$ is more close to $x$. \textbf{Second}, we propose a novel variant of EM algorithm called EMAS. It adds an adaptive smoothing step after the EM algorithm, whose smoothing weights are tweaked by considering the noise intensity and information of data structure. In contrast, SSW/SW uses fixed binomial coefficients to smooth the estimated distribution and details in the distribution are discarded.


\subsection{Randomizer Design}
In this subsection, we first introduce the motivation and intuition and then elaborate on the design of the randomizer $\mathcal{R}_{ASP}$ in ASP. We also assume the input domain is $[0, 1]$.


\noindent\textbf{Motivation and Intuition.}
The shortcomings of the baseline methods are attributed to that their randomizers do not take the advantage of the properties of the numerical domain or the shuffler.
To address these problems, the $\mathcal{R}_{ASP}$ inherits the idea of the squared-wave reporting, in which each user only outputs a single message while preserving numerical nature.
Since shuffle-DP only requires the privacy on the output of the shuffler, making the result satisfy $\epsilon_l$-LDP on the local side is an unnecessary constraint that narrows down the searching space for the optimal parameter. We observe that decomposing the privacy parameter can expand the searching space for the optimal parameter.
Based on the decomposed parameters, it is possible to apply a tighter mutual information bound to identify optimal parameters that allows the $\mathcal{R}_{ASP}$ to output a single message carrying more data information (high utility).

\begin{figure}[!tbp]
    \centering
    \includegraphics[scale=0.28]{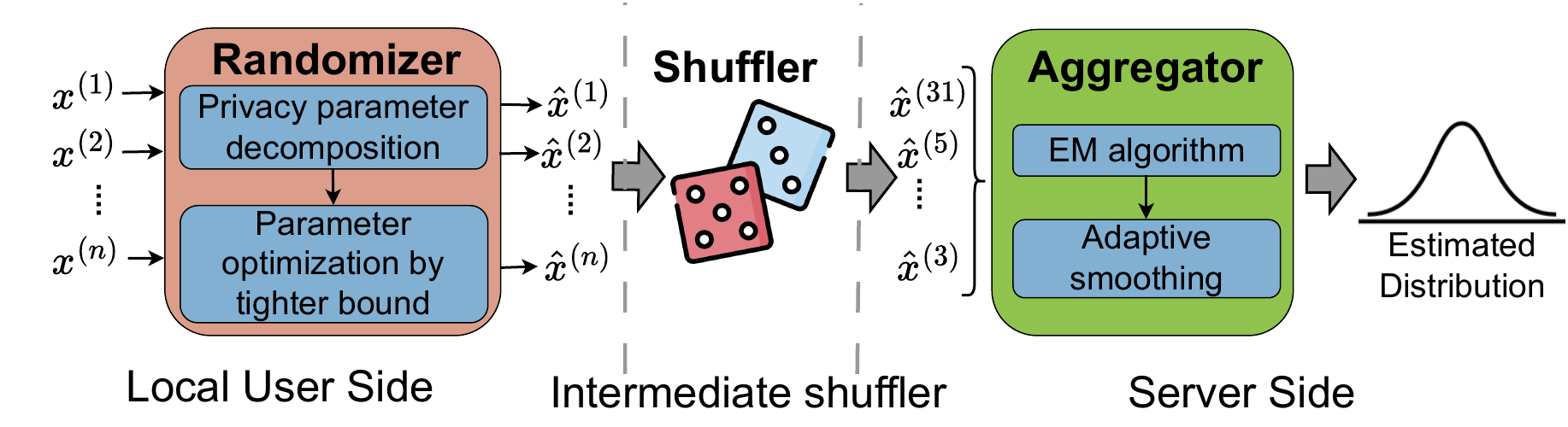}
    \vspace{-5pt}
    \caption{ASP Workflow}
    \vspace{-20pt}
    \label{fig:asp_workflow}
\end{figure}

\noindent\textbf{Design Detail.}
The goal of the $\mathcal{R}_{ASP}$ is to receive each user's data $x^{(i)}$ and output the perturbed result $\hat{x}^{(i)}$. The key component in $\mathcal{R}_{ASP}$ is the perturbation probability. As shown in the previous works~\cite{li2020estimating, wang2019collecting}, the square-wave perturbation probability guarantees the output distributions are the most distinguishable and enables the accurate distribution recovery, thereby we also employ this distribution in our $\mathcal{R}_{ASP}$.
Different from the baseline, we decompose the relationship between the $\epsilon_l$ and the perturbation parameters $p$ and $q$. Instead, we consider $p = kq (k > 1)$ and set $k$ and $b$ to be free variables ($p, q$ are determined by $k, b$ accordingly), which discards the unnecessary constraints and allows us to find the optimal parameters in a larger scope with more candidates.
Therefore, we define the $\mathcal{R}_{ASP}$ as follows.
\begin{align*}
    \Pr[\mathcal{R}_{ASP}(x) = \hat{x}] = 
    \begin{cases}
        p = \frac{k}{2bk + 1}, & \mathrm{if} \,\, |x - \hat{x}| = b \\
        q = \frac{1}{2bk + 1}, & \mathrm{otherwise}
    \end{cases}
\end{align*}
The most important step in our $\mathcal{R}_{ASP}$ is to choose the optimal parameters $k$ and $b$ that can achieve $(\epsilon, \delta)$-DP while preserving the most information in the perturbed result. In the rest of this section, we first analyze the conditions that the parameters should satisfy to achieve $(\epsilon, \delta)$-DP and then solve optimal parameters given the condition.


We also use privacy blanket to analyze the privacy guarantee that the parameter can provide.
Different from Theorem~\ref{the:privacy_ssw} for the baseline method that analyzes the relationship between $\epsilon_l$-LDP and $(\epsilon, \delta)$-DP, Theorem~\ref{the:privacy_asp} elaborates on the relationship between $b, k$ and $(\epsilon, \delta)$-DP.

\begin{theorem} \label{the:privacy_asp}
    The ASP is $(\epsilon, \delta)$-DP for any $b, k$ such that
    \begin{align*}
        \frac{(\xi(k-e^\epsilon-1+ke^\epsilon))^2}{4 \gamma n a} e^{-\gamma n (1 - e^{-\frac{2a^2}{(\xi(k-e^\epsilon-1+ke^\epsilon))^2}})} \leq \delta,
    \end{align*}
    where $\gamma = (1+2b)q, \xi = \frac{1}{(2bk + 1)(1+2b)}, a = e^\epsilon - 1$.
\end{theorem}


\begin{proof}
    The proof is similar to that of Theorem~\ref{the:privacy_ssw}.
    In ASP, we also have $\gamma = (1+2b)q$ and $\omega(y) = \frac{1}{1+2b}$. Therefore, we have
    \begin{align*}
        (1 - k e^\epsilon)\frac{q}{1+2b} \leq L \leq (k - e^\epsilon) \frac{q}{1+2b}.
    \end{align*}

    Thus, ASP satisfies $(\epsilon, \delta)$-shuffle DP if and only if the following inequality holds,
    \begin{align*}
        &\quad \frac{1}{\gamma n} \sum_{m=1}^{n} {n \choose m} \gamma^m (1 - \gamma)^{n-m} \\
        &\times \frac{(\frac{q}{1+2b}(k-e^\epsilon-1+ke^\epsilon))^2}{4a} e^{-\frac{2ma^2}{(\frac{q}{1+2b}(k-e^\epsilon-1+ke^\epsilon))^2}} \leq \delta,
    \end{align*}
    where $a = e^\epsilon - 1$. Using the binomial identity, we can simplify the above inequality and have 
    \begin{align*}
        \frac{(\frac{q}{1+2b}(k-e^\epsilon-1+ke^\epsilon))^2}{4 \gamma n a} e^{-\gamma n (1 - e^{-\frac{2a^2}{(\frac{q}{1+2b}(k-e^\epsilon-1+ke^\epsilon))^2}})} \leq \delta.
    \end{align*}
    
\end{proof}

\noindent\textbf{Parameter selection.}
Theorem~\ref{the:privacy_asp} delimits the feasible domain of the privacy parameter in ASP.
To find the optimal parameter that can maximize the utility, mutual information (MI) $I(X; \hat{X})$ is widely used in DP/LDP~\cite{li2020estimating, zhang2017privbayes, kairouz2016extremal} to quantify the correlation between the input $x \sim X$ and the output $\hat{x} \sim \hat{X}$, and higher correlation indicates more data information carried in $\hat{X}$ and better utility.
According to the definition, we have $I(X; \hat{X}) = h(\hat{X}) - h(\hat{X}|X)$ where the $h(\cdot)$ is the entropy.
Our baseline SSW and the prior work~\cite{li2020estimating} calculate the upper bound of $I(X; \hat{X})$ by simply assuming the output $\hat{X}$ is uniformly distributed.
However, this upper bound is unreachable for squared-wave reporting mechanism, and the parameter is far away from the optimal.
To explain the reason why the upper bound is unreachable, we demonstrate that the output $\hat{X}$ cannot be uniformly distributed over the domain $[-b, 1+b]$ for any continuous input distribution $p(x)$.
Specifically, we show that there exist two outputs $\hat{x}_1$ and $\hat{x}_2$ such that $\Pr[\mathcal{R}_{ASP}(X)=\hat{x}_1] \neq \Pr[\mathcal{R}_{ASP}(X)=\hat{x}_2]$.
Consider $\hat{x}_1 = -b$ and $\hat{x}_2 = -\frac{b}{2}$, we have $\Pr[\mathcal{R}_{ASP}(X)=\hat{x}_1] =q$, but $\Pr[\mathcal{R}_{ASP}(X)=\hat{x}_2] = \int_{0}^{\frac{b}{2}} p \times p(x)dx + \int_{\frac{b}{2}}^{1}q \times p(x)dx>q$.
We derive a tighter upper bound, combining which with Theorem~\ref{the:privacy_asp} we find the optimal parameter for $\mathcal{R}_{ASP}$ under $(\epsilon, \delta)$-DP.
We first find the input distribution for which the entropy $h(\hat{X})$ of the $\mathcal{R}_{ASP}$ output $\hat{X}$ is maximal.

\begin{theorem} \label{the:uniform_max_entropy}
    The entropy $h(\hat{X})$ of the output of $\mathcal{R}_{ASP}$ is maximum if the input $X$ follows the uniform distribution.
\end{theorem}



\begin{proof}

Consider an input $X$ over the domain $[0,1]$ with probability density $p_X(x)$, the output density $p_{\hat{X}}(\hat{x})$ is
\begin{align*}
p_{\hat{X}}(\hat{x}) = \int_0^1 p(\hat{x}|x) p_X(x) \, dx,
\end{align*}
and the entropy of $\hat{X}$ is
\begin{align*}
h(\hat{X}) = - \int_{\hat{X}} p_{\hat{X}}(\hat{x}) \log p_{\hat{X}}(\hat{x}) \, d\hat{x}.
\end{align*}

The goal is to maximize $h(\hat{X})$ over all input densities $p_X$,
\begin{align*}
    \begin{aligned}
        & \max_{p_X} h(\hat{X})
        & s.t.  \int_0^1 p_X(x) \, dx = 1, \quad p_X(x) \geq 0.
    \end{aligned}
\end{align*}

By adopting the Lagrangian functional, we have
\begin{align*}
\mathcal{L}[p_X, \lambda] = \int_{\hat{X}} p_{\hat{X}}(\hat{x}) \log \frac{1}{p_{\hat{X}}(\hat{x})} \, d\hat{x} + \lambda \left( \int_0^1 p_X(x) \, dx - 1 \right),
\end{align*}
where $\lambda$ is a Lagrange multiplier. The variational derivative is
\begin{align*}
\frac{\delta \mathcal{L}}{\delta p_X(x)} = - \int_{\hat{X}} p(\hat{x}|x) (\log p_{\hat{X}}(\hat{x}) + 1) d\hat{x} + \lambda.
\end{align*}

By setting the derivative to be zero for optimality, we have
\begin{align*}
\forall x \in [0, 1], \int_{\hat{X}} p(\hat{x}|x) (\log p_{\hat{X}}(\hat{x}) + 1) d\hat{x} = \lambda - 1,
\end{align*}
This implies that the left-side term is constant and independent of $x$.
Since $p(\hat{x}|x)$ is a constant determined by $\mathcal{R}_{ASP}$, the uniform distribution $p_X(x) = 1 (\forall x \in [0,1])$ maximizes the entropy $h(\hat{X})$.
\end{proof}

Based on the Theorem~\ref{the:uniform_max_entropy}, we can calculate the probability density function (PDF) of the output $\hat{X}$ given the uniformly distributed input $X$, and obtain the tighter upper bound of $I(X; \hat{X})$. Then, we first calculate the PDF of the output $\hat{X}$.
Since the input domain is $[0, 1]$ and the output domain is $[-b, 1+b]$, we divide the output domain into three segments: 

$\bullet$ \textbf{Segment $[-b, 0]$.} The probability density for $\hat{x} \in [-b, 0]$ is $\Pr[\hat{X} = \hat{x}] = \int_{0}^{1} p(\hat{x} | x) \cdot p(x) dx = p (\hat{x} + b) + q(1 - (\hat{x} + b))$.

$\bullet$ \textbf{Segment $[1, 1+b]$.} Symmetric to $[-b, 0]$, the probability density for $\hat{x} \in [1, 1+b]$ is $p\cdot (1+b - \hat{x}) + q(\hat{x} - 1)$.

$\bullet$ \textbf{Segment $[0, 1]$.} Since the summation of the probability density is one, the probability density for $\hat{x} \in [0, 1]$ is $\Pr[\hat{X} = \hat{x}] = 1 - (p - q)b^2 - 2qb$.

It is difficult to calculate the exact $h(\hat{X})$ based on the PDF of $\hat{X}$, because the value of the function over the segment $[-b, 0]$ and $[1, 1+b]$ varies with $\hat{x}$.
Therefore, we use a uniform distribution over the segment $[-b, 0]$ and $[1, 1+b]$ to approximate the actual PDF.
Since the uniform distribution has the constant density and larger entropy, this approximation will simplify the calculation and not shrink the upper bound.
Given the length $b$ of the segment $[-b, 0]$ and $[1, 1+b]$, the density of the uniform distribution is $\frac{\int_{-b}^{0} p\cdot (\hat{x}+b) + q(1 - (\hat{x} + b)) d\hat{x}}{b} = q + \frac{1}{2} (p-q)b$.
By the entropy definition, we have the tighter bound $h_u(\hat{X})$ of $h(\hat{X})$ is $h_u(\hat{X}) = -2[(qb + \frac{1}{2}(p-q)b^2) \log{(q + \frac{1}{2}(p-q)b)}] - (1 - (p - q)b^2 - 2qb) \log{(1 - (p - q)b^2 - 2qb)}$.

Since the $h(\hat{X} | X)$ only depends on the $\mathcal{R}_{ASP}$, the upper bound $I_{u}(X; \hat{X})$ of $I(X;\hat{X})$ is $I_{u}(X; \hat{X}) = h_u(\hat{X}) + 2bp \log{p} +q\log{q}$. By Theorem~\ref{the:privacy_asp}, we choose optimal parameters for $\mathcal{R}_{ASP}$ by solving the following optimization problem.
\begin{align*}
    & \max\limits_{b, k} \quad\quad\quad\quad\quad\quad I_{u}(X; \hat{X}) \\
    & \mathrm{s.t.} \frac{(\xi(k-e^\epsilon-1+ke^\epsilon))^2}{4 \gamma n a} e^{-\gamma n (1 - e^{-\frac{2a^2}{(\xi(k-e^\epsilon-1+ke^\epsilon))^2}})} \leq \delta
\end{align*}


We have attempted several methods to find the accurate solution for the above problem. The first attempt is to convert this problem to an unconstrained optimization problem by Courant-Beltrami penalty function~\cite{boyd2004convex}, and then adopt methods for unconstrained optimization problems, such as BFGS~\cite{boyd2004convex} and Powell~\cite{powell1964efficient}, to find $b$ and $k$. However, due to the complex nonlinear constraint, the solution always deviates much from the optimal value and degrades the utility. Therefore, we alternatively adopt the SLSQP algorithm~\cite{ma2024improved}, since it effectively handles nonlinear constraints and, in our experiments, consistently finds accurate values of $b$ and $k$.

\begin{figure}
    \centering
    \includegraphics[scale=0.3]{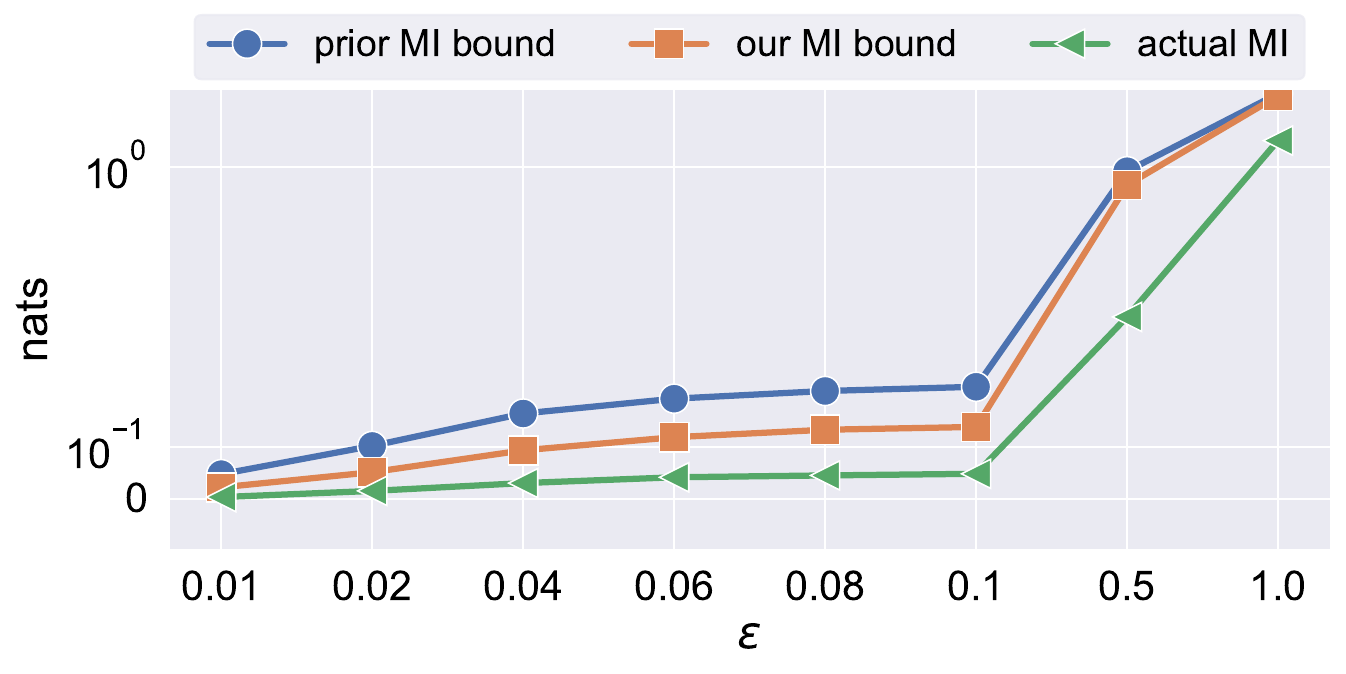}
    \vspace{-10pt}
    \caption{Our MI upper bound vs. Prior MI upper bound on \textit{Normal} dataset with $\delta = 10^{-5}$.}
    \label{fig:MI_bound_comparison}
    \vspace{-18pt}
\end{figure}

\noindent\textbf{Verify the Tighter Upper Bound.}
Figure~\ref{fig:MI_bound_comparison} illustrates the comparison of our bound against both the actual MI and the prior bound on \textit{Normal} dataset used in our experiments.
We observe that our upper bound is about 50\% tighter than the prior bound.
Therefore, the above evaluation shows the capability of our tighter bound in finding better parameters.

\noindent\textbf{Difference between ASP and SSW.}
ASP has two primary differences compared to SSW.


\begin{itemize}[leftmargin=*]
    \item \textbf{Tunable Parameters vs. Fixed Parameters.} SSW uses a fixed $\epsilon_l$ to satisfy $(\epsilon, \delta)$-DP after shuffling, which limits the search space of the optimal parameter.
    Whereas, ASP is built upon two free parameters, which enables us to find optimal parameters satisfying $(\epsilon, \delta)$-DP in a larger space.

    \item \textbf{Tighter MI Bound vs Loose MI Bound.} SSW uses the unreachable upper bound of $I(X; \hat{X})$ under the false assumption that the output is uniformly distributed.
    In the contrast, we derive a tighter upper bound to facilitate the optimal parameter selection in ASP.
\end{itemize}

\subsection{Aggregation Design}
We first introduce the issues of existing EM variants to motivate our design, and then elaborate on the design details.

\noindent\textbf{Motivation.}
Variants of the EM algorithm are widely used in differential privacy protocols~\cite{li2020estimating, yerevisiting, erlingsson2014rappor}.
It receives $n$ perturbed values and outputs a histogram with $m$ bins $\bm{\tilde{f}}=[\tilde{f}_1, ..., \tilde{f}_m]$ as the estimated distribution.
However, the state-of-the-art variants, EMS~\cite{li2020estimating} and MR~\cite{yerevisiting}, usually implicitly requires that the data has some specific characteristics for a favorable performance.
Specifically, EMS uses fixed weights to average adjacent bins under the assumption that the distribution is always smooth. Although the indiscriminative smoothing enhances the robustness, it reduces utility for jagged distributions.
MR uses reduction technique to merge bins with small and similar frequencies but does not modify high-frequency bins. It fails to resist poisoning attacks to promote densities.
Besides, once the distribution does not have many bins with small frequencies, MR degenerates to the traditional EM algorithm without any accuracy improvement.


\noindent\textbf{Intuition.}
EMAS enhance the robustness and improves the utility by executing an additional adaptive smoothing (AS-step) after the E-step and M-step of the EM algorithm.
To adaptively smooth the $i$-th estimate $\tilde{f}_i$, EMAS conducts dynamic weighted average to the $\tilde{f}_j$ near the $\tilde{f}_i$.
EMAS resorts to three types of information to dynamically tweak the weight.
\begin{enumerate}[leftmargin=*]
    \item \textbf{Frequency difference between $\tilde{f}_i$ and $\tilde{f}_j$.} Larger difference $|\tilde{f}_i - \tilde{f}_j|$ indicates that their actual frequencies tend to be clearly distinct. Therefore, $\tilde{f}_j$ should have less weight.

    \item \textbf{Position difference between $\tilde{f}_i$ and $\tilde{f}_j$.} Larger position distance $ |i - j| $ between $\tilde{f}_i$ and $\tilde{f}_j$ means lower similarity between them, thereby it is reasonable to assign less weight to $\tilde{f}_j$ to alleviate the negative impact of $\tilde{f}_j$ to $\tilde{f}_i$.

    \item \textbf{The characteristic of EM algorithm.} During the EM iteration, the improvement rate of the likelihood value is rapid at the initial stage but slows down in the later iteration due to mixture data components~\cite{nychka1990some, yerevisiting}.
    Therefore, fixed average weights could make the estimated result oscillate around the suboptimal point.
    To avoid this problem, we comply the feature of EM and leverage the weight decay technique~\cite{d2024we}, which is a promising methodology to tackle the similar problem~\cite{ge2019step, loshchilov2017sgdr}.


\end{enumerate}


\begin{algorithm}[!tbp]
\small
\caption{Aggregation of EMAS} \label{alg:EMAS}
\KwIn{Transition matrix $M$, noisy report $[\hat{x}^{(i)} ..., \hat{x}^{(n)}]$, maximum number of iteration $\tau$, average radius $R$}
\KwOut{Estimated distribution $[\tilde{f}_1, ..., \tilde{f}_m]$}
    Round $t \leftarrow 0$ \;
    Initialize $\tilde{\bm{f}}^0 = [\frac{1}{m}, ..., \frac{1}{m}]$ \;
    \While{not converge and $t \leq \tau$} {
        E-step: $\forall i \in \{1, 2, ..., m\}$
        \begin{align*}
            P_i = \tilde{f}_i \sum_{j \in [\hat{m}]} n_j \frac{M_{j, i}}{\sum_{k=1}^{m} M_{j, k} \tilde{f}_k}
        \end{align*} \\
        M-step: $\forall i \in \{1, 2, ..., m\}$
        \begin{align*}
            \tilde{f}_i = \frac{P_i}{\sum_{k=1}^{m} P_k}
        \end{align*} \\
        AS-step: $\forall i \in \{1, 2, ..., m\}$ \\
        $\sigma_1 \leftarrow \frac{1}{m}\sum_{j=1}^{m}\sum_{i=1}^{m} n_i \frac{m^2 M_{i, j}}{\sum_{s=1}^{m} M_{i, s}}$ \;
        $\sigma_2 \leftarrow \frac{1}{3} + \frac{1}{3} \times \left( 1 - \cos{(\pi \times \frac{t}{50})} \right)$ \;
        \ForEach{$k \in [i-R, i+R]$} {
            $w_k^t \leftarrow \mathcal{K}_{\sigma_1}(\tilde{f}_i^t – \tilde{f}_k^t) \times \mathcal{K}_{\sigma_2}(i-k)$ \;
        }
        $\tilde{f}_i^t \leftarrow \sum_{k=i-R}^{i+R} w_k^t f_k^t$ \;
    }
    \KwRet $[\tilde{f}_1, ..., \tilde{f}_m]$ \;
\end{algorithm}

\noindent\textbf{Design Detail.}
Algorithm~\ref{alg:EMAS} describes our EMAS. EMAS performs the reconstruction in a discrete manner. It first receives $n$ users' reports in the domain $\hat{D} = [-b, 1+b]$ and discretizes them into $\hat{m}$ bins.
Then, EMAS iteratively executes expectation (E-step), maximization (M-step) and the proposed adaptive smoothing (AS-step) to output a histogram with $m$ bins $[\tilde{f}_1, ..., \tilde{f}_m]$ as the estimated distribution.
For simplicity, we set $\hat{m} = m$ in our experiments.
Let $n_j$ be the number of perturbed values in the $j$-th bin $\hat{B}_j$ in the output domain, and denote the ASP perturbation as a transition matrix $M \in [0, 1]^{\hat{m} \times m}$ where $M_{j, i}$ is the probability of the perturbed value falling in the bin $\hat{B}_j$, given the input in the bin $B_i$, the E-step calculates $P_i = \tilde{f}_i^t \sum_{j \in [m]} n_j \frac{M_{j, i}}{\sum_{k=1}^{m} M_{j, k} \tilde{f}_k^t} (\forall i \in [m])$ for each $\tilde{f}_i^t$ in the $t$-th iteration, and then we have $Q(\tilde{\bm{f}} | \tilde{\bm{f}}^t) = \sum_{i=1}^{m} P_i \ln{\tilde{f}_i}$.
By making the derivative of $Q(\tilde{\bm{f}} | \tilde{\bm{f}}^t)$ be zero, the M-step yields $\tilde{f}_i^{t} = \frac{P_i}{\sum_{j=1}^{m} P_j} (\forall i \in [m])$.

We then elaborate on our AS-step.
In the $t$-th iteration, we have the estimate $\tilde{f}_1^t, ..., \tilde{f}_m^t$ after the M-step, and we conduct weight average on each $\tilde{f}_i^t$ and the bins near it to achieve the adaptive smoothing. Formally, we have $\tilde{f}_i^t = \sum_{k=i-R}^{i+R} w_k^t f_k^t$,
where $w_k^t$ is the weight and $R$ is the average radius.
Based on the above intuition, each weight consists of three components: 1) \textit{frequency difference weight} $\mathcal{K}_{\sigma_1}(\tilde{f}_i^t - \tilde{f}_k^t)$, 2) \textit{position difference weight} $\mathcal{K}_{\sigma_2}(i-k)$ and 3) weight decay function $\mathcal{E}(t)$ with respect to the iteration round $t$. Thus, we have $w_k^t = \mathcal{K}_{\sigma_1}(\tilde{f}_i^t - \tilde{f}_k^t) \times \mathcal{K}_{\sigma_2 = \mathcal{E}(t)}(i-k)$,
where $\mathcal{K}_{\sigma_i}(x) = \frac{1}{\sigma_i\sqrt{2 \pi}} e^{-\frac{x^2}{2\sigma_i^2}}$ is the Gaussian kernel function~\cite{silvey2017statistical} mapping the frequency/position difference to a numerical weight.

Since Gaussian function is a bell-shaped curve, it can provision less weight if the difference is larger and thereby substantiates the intuition (1) and (2).
The weight decay function $\mathcal{E}(t)$ dynamically adjusts the weight assigned by Gaussian function and thus achieve the intuition (3).
Next, we describe the unexplained parameters $\sigma_1$ and $\sigma_2$.
We first discuss the frequency difference weight $\mathcal{K}_{\sigma_1}(\tilde{f}_i^t - \tilde{f}_k^t)$. In Gaussian function, 99.7\% of the values lie within $3\sigma_1$ of the mean. Because of this, $\tilde{f}_k$ will be ignored by assigning a very small weight if the frequency difference $|\tilde{f}_i^t - \tilde{f}_k^t|$ exceeds $3\sigma_1$.
Therefore, it is a good idea to set the $\sigma_1$ to be the standard deviation of the estimate under the EM algorithm.
In this way, for any bins whose actual frequencies are close, their perturbed frequencies should fall in $[\tilde{f}_i^t - 3\sigma_1, \tilde{f}_i^t + 3\sigma_1]$ with high probability.
Although EM algorithm has no closed-form expression for the estimate and it is challenging to solve the rigorous standard deviation, we employ the Cramer-Rao lower bound~\cite{rice2002cramer} to approximate the $\sigma_1$ via Fisher information, and a reference value for $\sigma_1$ is finally obtained as $\frac{1}{m}\sum_{j=1}^{m}\sum_{i=1}^{m} n_i \frac{m^2 M_{i, j}}{\sum_{s=1}^{m} M_{i, s}}$.

For the position difference weight $\mathcal{K}_{\sigma_2 = \mathcal{E}(t)}(i-k)$, the bins lying within $3\sigma_2$ of $i$ are involved into average with non-negligible weights.
Therefore, $\sigma_2$ substantially controls the size of the smoothing window, and larger $\sigma_2$ provides wider window.
To achieve a consistent superior performance across datasets, we empirically use the \textit{cosine decay}~\cite{d2024we, loshchilov2017sgdr} to periodically update $\sigma_2$ and have $\sigma_2 = \mathcal{E}(t) = \sigma_{min} + \frac{1}{2}(\sigma_{max} - \sigma_{min}) \times (1 - \cos{(\pi \times \frac{t}{\mathcal{T}}}))$, where $\sigma_{min}$ and $\sigma_{max}$ are the minimum and the maximum of $\sigma_2$, and $\mathcal{T}$ is the length of the update period. The cosine decay is advantageous because the period variation is more reliable than other decay schemes and is more capable of tackling mixture data components.
As shown in~\cite{li2020estimating}, the smoothing performs satisfactorily when the length of the smoothing window is greater than 3, but an excessively large size (e.g, 7) will disrupt the detailed distribution information such as the spiky part.
Therefore, we set the minimum and the maximum of $3\sigma_2$ to be $\frac{3-1}{2}=1$ and $\frac{7-1}{2}=3$, and thus $\sigma_{min} = \frac{1}{3}$ and $\sigma_{max} = 1$.
Since smoothing in the early stage of each cycle tends to blur distribution details, we set the smoothing window to be small at the first stage to preserve distribution details, and to be wide at the final stage to better polish the overall distribution shape.
The length $\mathcal{T}$ is a trade-off between two errors. Small $\mathcal{T}$ leads to rapid change of $\sigma_2$ and the estimate would oscillate around the optimal result, while large $\mathcal{T}$ reduces the effectiveness of the weight decay and degrades the utility. We empirically study the impact of different values (see Section~\ref{sec:utility_evaluation}) and choose $\mathcal{T} = 50$ for better utility.

\noindent\textbf{Stopping Criterion.}
One of the most general stopping criterion is checking whether the relative improvement between two iterations is small~\cite{erlingsson2014rappor}. During the iterations of EMAS, the $L_1$ distance between two different estimated distributions is at least $\frac{1}{n}$ for $n$ users. We therefore consider that the EMAS has converged if $\| \tilde{\bm{f}}^{t+1} - \tilde{\bm{f}}^{t} \|_1 < \frac{1}{n}$. To prevent the EMAS from failing to terminate, we also set the maximum number of iterations to $\tau = 10000$.
Although the main convergence condition is the $L_1$ bound, we also test the convergence time under $\tau = 5000, 8000, 10000, 12000$.
We observe that EMAS only costs about 20s when $n = 100000$ for varying $\tau$, which is often tolerable in real-world applications.


\subsection{Utility Analysis and Potential Extension}

\noindent\textbf{Utility Discussion.}
Although lower bound of error is a standard for utility analysis, comparing it for ASP is non-trivial. Existing shuffle-DP protocols~\cite{balle2020private, ghazi2021differentially} for numerical data focus on the summation task instead of the distribution estimation, and thus direct comparison is not informative. Moreover, although basic SCFOs for frequency estimation have theoretical lower bound of error, analyzing the lower bound of error for our ASP and SCFO with binning and consistency is challenging since the EMAS and consistency methods lack closed-form expression that would allow a tight, tractable lower-bound analysis.

Therefore, we study the convergence of our EMAS to show how ASP approaches the optimal states.
Our analysis is built upon the established result~\cite{loshchilov2017sgdr} that cyclic weights improve the convergence rate in stochastic optimization.
Therefore, we analyze the convergence of a degraded version of EMAS whose smoothing weights are fixed.
The result can still shed light on the utility of EMAS, and the empirical study in Section~\ref{sec:utility_evaluation} bridges the gap between theory and practice and confirms the better utility of EMAS.
Formally, the analysis result is shown in Theorem~\ref{the:convergence}.
We first introduce some useful notations used in the theorem.
We denote each iteration of the degraded EMAS as the function $\mathbb{M}(\cdot)=S (EM( \cdot ))$ where $S$ is the smoothing matrix, denote the estimated frequencies before the AS-step as $\hat{\bm{f}}$ and thus we have $\tilde{\bm{f}} = S(\hat{\bm{f}})$, use $\tilde{\bm{f}}^*$ and $P_i^*$ to denote the estimate and the value of $P_i$ when the algorithm has converged. 

\begin{theorem}\label{the:convergence}
    The degraded EMAS will converge to $\tilde{\bm{f}}^*$ which can maximize the penalized likelihood $L(\bm{f}) - \sqrt{\bm{f}}^T \mathcal{R} \sqrt{\bm{f}}$ with the convergence rate upper bounded by a geometric rate of $\Gamma = \| \frac{\partial}{\partial \bm{f}} \mathbb{M}(\tilde{\bm{f}}^*) \|$, where $\mathcal{R}$ is the penalized matrix, if $\Gamma < 1$, the $i$-th column $\mathcal{R}_i$ of $\mathcal{R}$ satisfies $\mathcal{R}_i^T \sqrt{\tilde{\bm{f}}^*}=\frac{P^*_i}{\sqrt{\tilde{f}_i^*}}$.
\end{theorem}


\begin{proof}
    Our proof consists of two steps. We first prove that the convergence point can achieve the maximum penalized likelihood $L(\bm{f}) - \sqrt{\bm{f}}^T \mathcal{R} \sqrt{\bm{f}}$, and then prove the upper bound of the convergence rate.

    We compute the derivative of $L(\bm{f}) - \sqrt{\bm{f}}^T \mathcal{R} \sqrt{\bm{f}}$,
    \begin{align*}
        &\quad\, \frac{\partial \left( L(\bm{f}) - \sqrt{\bm{f}}^T \mathcal{R} \sqrt{\bm{f}} \right)}{\partial \sqrt{f_i}} = \frac{2 P_i}{\sqrt{f_i}} - 2\mathcal{R}_i^T \sqrt{\bm{f}}.
    \end{align*}
    To find the convergence point $\tilde{\bm{f}}^*$ such that the penalized likelihood is maximum, the derivative on $\tilde{\bm{f}}^*$ should be zero,
    \begin{align*}
         \frac{2 P_i^*}{\sqrt{\tilde{f}_i^*}} - 2\mathcal{R}_i^T \sqrt{\tilde{\bm{f}}^*} = 0
    \end{align*}

    Then we prove the upper bound of the convergence rate. Since $\tilde{\bm{f}}^*$ is the convergence point, we have $\mathbb{M}(\tilde{\bm{f}}^*) = \tilde{\bm{f}}^*$, we need to relate the distance between $\bm{f}^0$ and $\tilde{\bm{f}}^*$ to the improvement when $\mathbb{M}$ is applied to $\bm{f}^0$.
    \begin{align*}
        \mathbb{M}(\tilde{\bm{f}}^{0}) - \tilde{\bm{f}}^* &= \mathbb{M}(\tilde{\bm{f}}^{0}) - \mathbb{M}(\tilde{\bm{f}}^*) \\
        &= \frac{\partial}{\partial \bm{f}} \mathbb{M}(\tilde{\bm{f}}^*) (\tilde{\bm{f}}^0 - \tilde{\bm{f}}^*) + \frac{\partial^2}{\partial \bm{f}^2} \mathbb{M}(\xi) (\tilde{\bm{f}}^0 - \tilde{\bm{f}}^*)^2
    \end{align*}
    where $\xi \in [\tilde{\bm{f}^0}, \tilde{\bm{f}^*}]$. Since the norm of $\frac{\partial^2}{\partial \bm{f}^2} \mathbb{M}(\xi)$ is bounded, there exist $K > 0$ such that
    \begin{align*}
        \| \tilde{\bm{f}}^{1} - \tilde{\bm{f}}^* \| = \| \mathbb{M}(\tilde{\bm{f}}^{0}) - \tilde{\bm{f}}^* \| \leq \Gamma \| \tilde{\bm{f}}^0 - \tilde{\bm{f}}^* \| + K \| \tilde{\bm{f}}^0 - \tilde{\bm{f}}^* \|^2
    \end{align*}
    By the similar calculation, we can calculate $\| \tilde{\bm{f}}^{2} - \tilde{\bm{f}}^* \|$,
    \begin{align*}
        &\| \tilde{\bm{f}}^{2} - \tilde{\bm{f}}^* \| = \| \mathbb{M}(\tilde{\bm{f}}^{1}) - \tilde{\bm{f}}^* \| \leq \Gamma \| \tilde{\bm{f}}^1 - \tilde{\bm{f}}^* \| + K \| \tilde{\bm{f}}^1 - \tilde{\bm{f}}^* \|^2 \\
        &\leq \Gamma^2 \| \tilde{\bm{f}}^0 - \tilde{\bm{f}}^* \| + \Gamma K \| \tilde{\bm{f}}^0 - \tilde{\bm{f}}^* \|^2 + K\Gamma \| \tilde{\bm{f}}^0 - \tilde{\bm{f}}^* \| \\
        & \quad + K^2 \| \tilde{\bm{f}}^0 - \tilde{\bm{f}}^* \|^2
    \end{align*}
    Therefore, it is simple to show the convergence rate by induction. Since $\Gamma$ depends on the $\tilde{\bm{f}}^*$ and prior work~\cite{nychka1990some} has demonstrated that $\Gamma$ could be less than 1, when $\Gamma < 1$ and $K \| \tilde{\bm{f}}^0 - \tilde{\bm{f}}^* \| < 1$, the convergence rate achieves the maximal, and the $\tilde{\bm{f}}^i$ will convergence to $\tilde{\bm{f}}^*$ with a geometric rate $\Gamma$.
\end{proof}

\noindent\textbf{Extension to High-Dimensional Data.}
ASP could be extended to high-dimensional data by the canonical budget split technique. We split the total budget $\epsilon$ into $d$ dimensions to estimate each individual attribute, and the high-dimensional distribution can be estimated by the product of the associating attributes. Our experiments show that the direct extension does work well for all situations and highlights the need for sophisticated protocols.

\noindent\textbf{MI-based Tuning for Protocol Design.}
In addition to parameter selection, MI as a correlation indicator between the input and the output of a protocol can also be adopted to design the optimal perturbation function for other tasks for better utility. One may use the characteristic of the task and design a perturber that maximizes the correlation between the input and output by mutual information.

\noindent\textbf{Building Block for LSH Kernel Density Estimation.}
ASP can serve as a building block in a differentially private LSH kernel density estimation (DP-LSH-KDE). In DP-LSH-KDE, each user hashes their data using $L$ LSH functions and reports randomized results; the server then reconstructs the distribution of these $L$ values to estimate the LSH-KDE. By incorporating ASP, the server can more accurately recover this distribution, leading to improved estimation quality.


%% file: utility_evaluation_v2.tex
\section{Utility Evaluation} \label{sec:utility_evaluation}

\begin{figure}[!tbp]
    \centering
    \includegraphics[scale = 0.3]{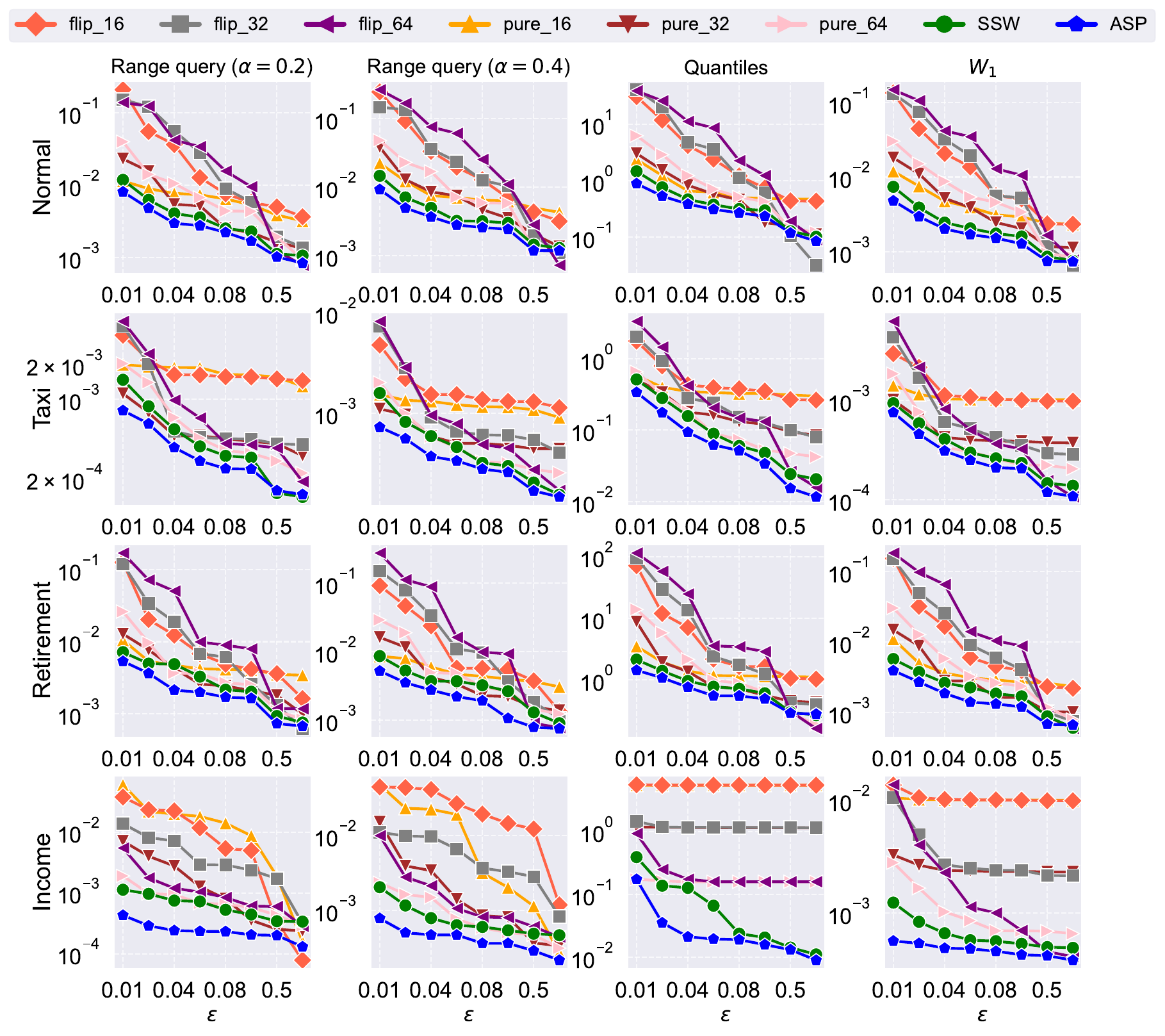}
    \vspace{-10pt}
    \caption{Utility results with varying $\epsilon$. Each row corresponds to one dataset.}
    \vspace{-15pt}
    \label{fig:utility_with_eps}
\end{figure}

\subsection{Setup}

\noindent\textbf{Datasets.}
We use one synthetic dataset and three real-world datasets to conduct our experiments.
\begin{itemize}[leftmargin = *]
    \item \textit{Synthetic Normal Dataset.} We drew $10^5$ samples from the Normal distribution $\mathcal{N}(0, 10)$ to form the dataset. We choose to set the estimated histograms with 256 bins.

    \item \textit{Taxi~\cite{taxi}.} This dataset was published by the New York Taxi Commission in 2018. It contains 2,189,968 samples of pickup time in a day (in seconds). We choose to set the estimated histograms with 512 bins.

    \item \textit{Retirement~\cite{retirement}.} This dataset contains data about San Francisco employee retirement plans, ranging from -28,700 to 101,000. We extract non-negative values smaller than 60,000 for evaluation. We choose to set the estimated histograms with 512 bins.

    \item \textit{Income~\cite{income}.} We use the personal income information of the 2017 American Community Survey [26]. We extract the values that are smaller than 524288 for evaluation. We choose to set the estimated histograms with 512 bins.

\end{itemize}

\noindent\textbf{Competitors.}
We compare our ASP with the baseline protocols proposed in Section~\ref{sec:baseline_method}: Flip, Pure and SSW.
For Flip and Pure, we partition the input domain into $c$ consecutive, non-overlapping chunks. We consider $c = 16, 32, 64$ and set $\hat{m} = m = 512$ for SSW and ASP, which are the best performing values in most cases~\cite{li2020estimating}.

\noindent\textbf{Metric.}
We evaluate the performance by three statistical tasks. The task and the corresponding metric is summarized below.
\begin{itemize}[leftmargin=*]
    \item \textbf{Range Query.} Define a range query function $\bm{R}(\bm{f}, i, \alpha) = \bm{P}(\bm{f}, i+\alpha) - \bm{P}(\bm{f}, i)$, where $\bm{P}(\bm{f}, i)=\sum_{v=0}^{i}\bm{f}_v$ is the cumulative distribution function and $\alpha$ is the range size. We randomly sample different $i$ to reflect the estimate quality by calculating $|\bm{R}(\bm{f}, i, \alpha) - \bm{R}(\tilde{\bm{f}}, i, \alpha)|$. We set $\alpha = 0.2$ and 0.4 for a comprehensive evaluation with different sizes.

    \item \textbf{Quantiles.} Quantiles are cut points dividing the distribution into intervals with equal probabilities. Formally, $\beta$-quantile is $\bm{Q}(\bm{f}, \lambda) = \arg\max_{i} \bm{P}(\bm{f}, i) \leq \lambda$. In our experiment, we define $\mathcal{Q}=\{ 5\%, 10\%, ..., 95\% \}$ and measure $\frac{1}{|\mathcal{Q}|} \sum_{\lambda \in \mathcal{Q}} |\bm{Q}(\tilde{\bm{f}}, \lambda) - \bm{Q}(\bm{f}, \lambda)|$.

    \item \textbf{Wasserstein Distance.} We use Wasserstein distance $W_1(\bm{f}, \tilde{\bm{f}})$ to measure the overall closeness between two distributions. Formally, $W_1(\bm{f}, \tilde{\bm{f}}) = \sum_{i=1}^{m} |\bm{P}(\bm{f}, i) - \bm{P}(\tilde{\bm{f}}, i)|$.
\end{itemize}

All metrics measure the error between the estimate with the ground truth, thereby smaller values indicate better utility.

\noindent\textbf{Parameter setting.}
We vary $\epsilon$ from 0.01 to 1 to cover a wide range of privacy level, and set default $\delta = 10^{-5}$ and $R = 3$.
We repeat the experiments on each dataset and method 100 times and record the average result.

\begin{figure}[!htbp]
    \centering
    \includegraphics[scale = 0.3]{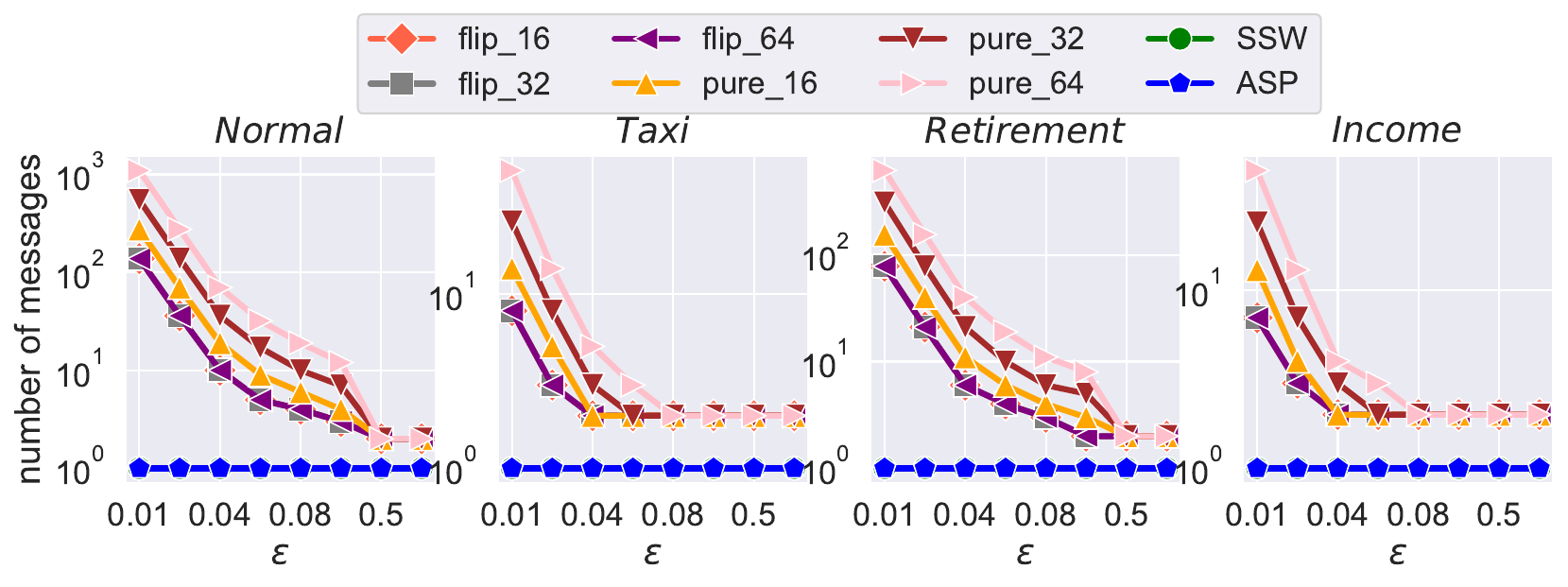}
    \vspace{-10pt}
    \caption{Message complexity result with varying $\epsilon$.}
    \label{fig:message_with_eps}
\end{figure}

\begin{figure}[!htbp]
    \centering
    \includegraphics[scale = 0.3]{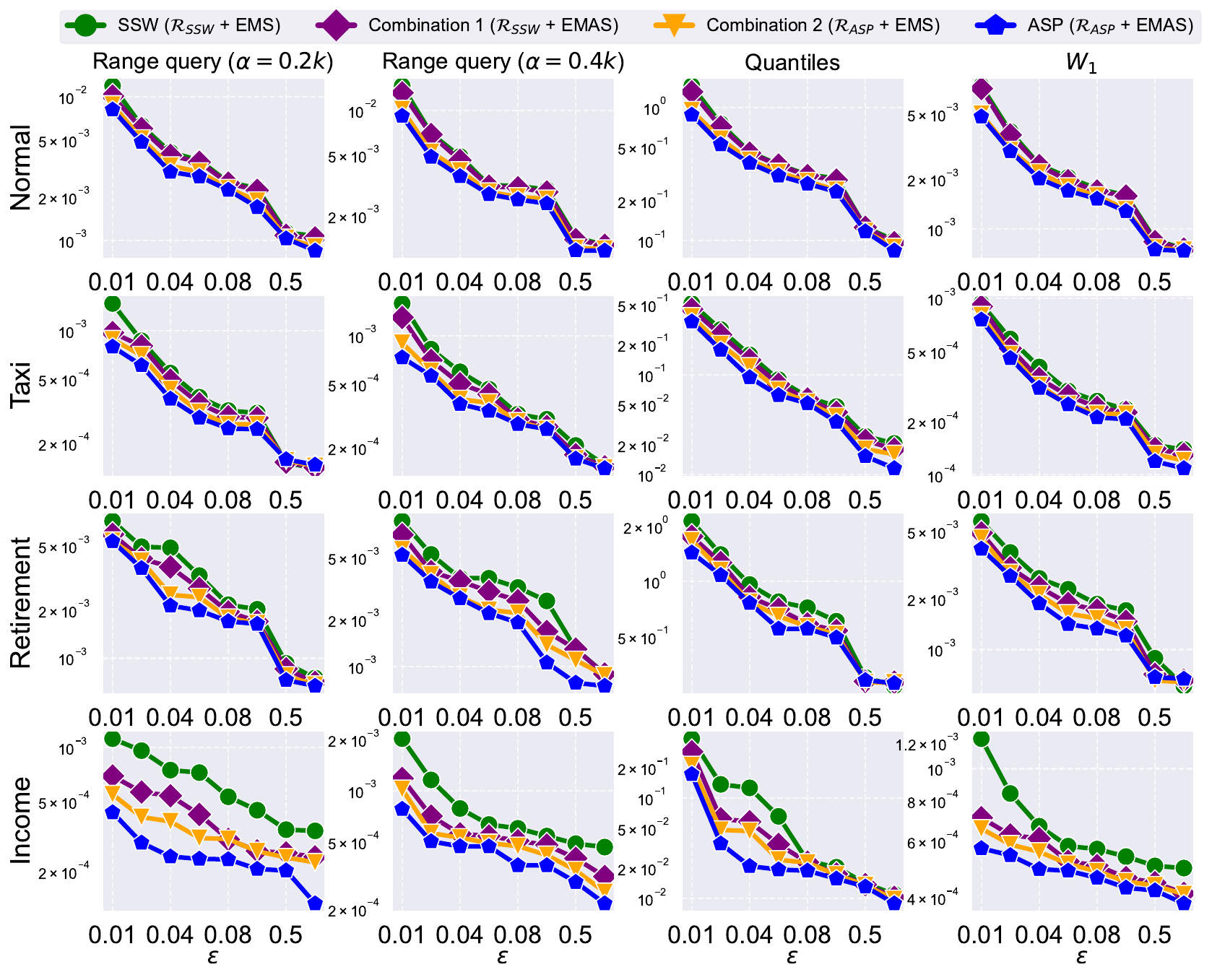}
    \vspace{-5pt}
    \caption{Component study.}
    \vspace{-15pt}
    \label{fig:ablation_main_result}
\end{figure}

\subsection{Overall Results}

\noindent\textbf{Protocol Comparison.}
The comparison in terms of utility and message complexity is shown in Figure~\ref{fig:utility_with_eps} and \ref{fig:message_with_eps}. Key observations are listed below.

\begin{itemize}[leftmargin=*]
    \item ASP outperforms baseline methods on all statistical tasks.
    \begin{itemize}
        \item The \textit{Normal}, \textit{Taxi} and \textit{Retirement} have smooth distributions, and our ASP performs better on these datasets because it can search the optimal perturbation parameter.

        \item ASP has a significant improvement on \textit{Income} which is spiky than other datasets. Beyond the effectiveness of the optimal parameter, the improvement predominantly attributes to our adaptive smoothing.
        
    \end{itemize}

    \item For range queries, larger $\alpha$ accumulates errors in more bins and shows lower accuracy.

    \item For SCFO-based methods, the partition $c$ is critical and there is no one-size-fit-all parameter $c$ for all tasks, since $c$ has a trade-off between noise variance and bias.

    \item As shown in Figure~\ref{fig:message_with_eps}, ASP and SSW have the lowest message complexity due to their single-message randomizer.
    

    \item By applying dummy point technique, Pure and Flip have large message complexity especially under small $\epsilon$ because they need more dummy points for higher privacy demand.
\end{itemize}

\subsection{Ablation Study}
In this section, we first study the contribution of $\mathcal{R}_{ASP}$ and the EMAS, and then study the impact of different hyper-parameters, including $R$, $\sigma_2$ and the size of dataset.
Unless otherwise mentioned, we only use \textit{Income} datasets for hyper-parameter evaluation.

\noindent\textbf{Component Study.}
We empirically study the effectiveness of the $\mathcal{R}_{ASP}$ and EMAS to further explore the contribution of each ASP component to the overall performance.
Specifically, we consider two combined protocols: (1) \textit{Combination 1} is the randomizer $\mathcal{R}_{SSW}$ with the EMAS in ASP and (2) \textit{Combination 2} is our $\mathcal{R}_{ASP}$ with traditional EMS.
We compare these combined protocols with ASP to show how each component affects the performance.
As shown in Figure~\ref{fig:ablation_main_result}, both $\mathcal{R}_{ASP}$ and EMAS contribute to ASP performance and $\mathcal{R}_{ASP}$ contributes more than EMAS, because it optimizes the essential parameter.
On \textit{Normal}, \textit{Taxi} and \textit{Retirement}, the performance of \textit{Combination} 1 is close to SSW, since the constant-weighted smoother is able to handle non-spiky distribution.
While on \textit{Income}, both $\mathcal{R}_{ASP}$ and EMAS benefit the estimate since the optimal parameter and adaptive smoothing in EMAS help to preserve spiky details in distribution compared to EMS.

\noindent\textbf{Impact of Hyper-parameters.}
\begin{figure}[!tbp]
    \centering
    \includegraphics[scale = 0.3]{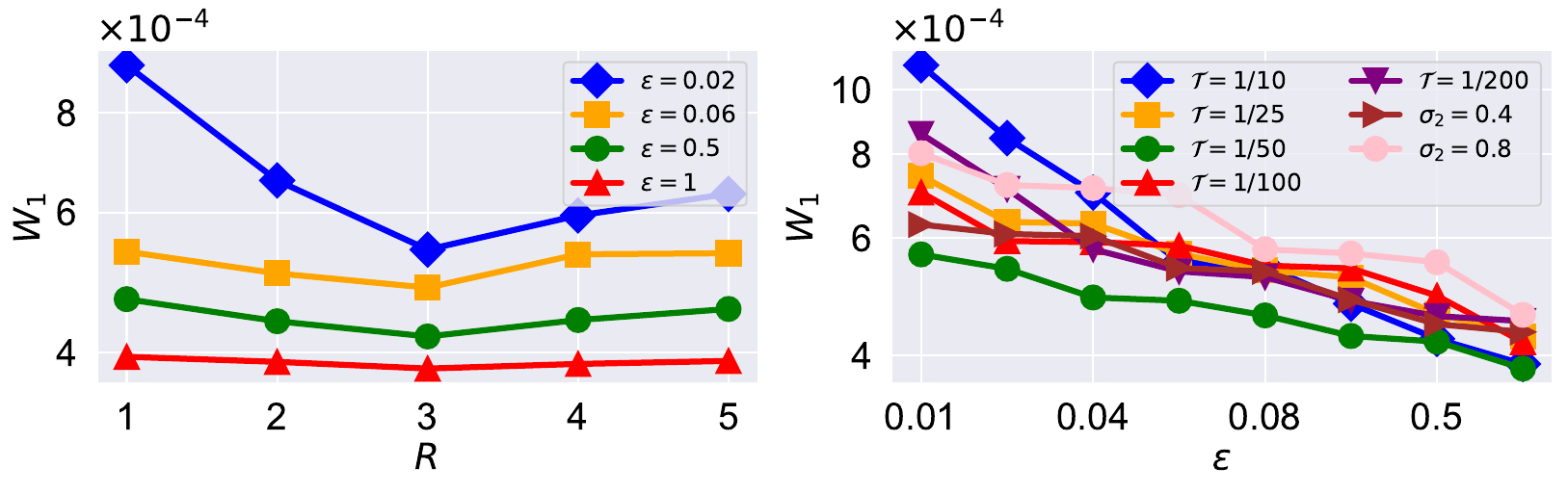}
    \vspace{-10pt}
    \caption{$W_1$ distance with varying hyper-parameters on \textit{Income}.}
    \label{fig:hyperparameter_test}
    \vspace{-15pt}
\end{figure}
\begin{figure}
    \centering
    \includegraphics[scale = 0.29]{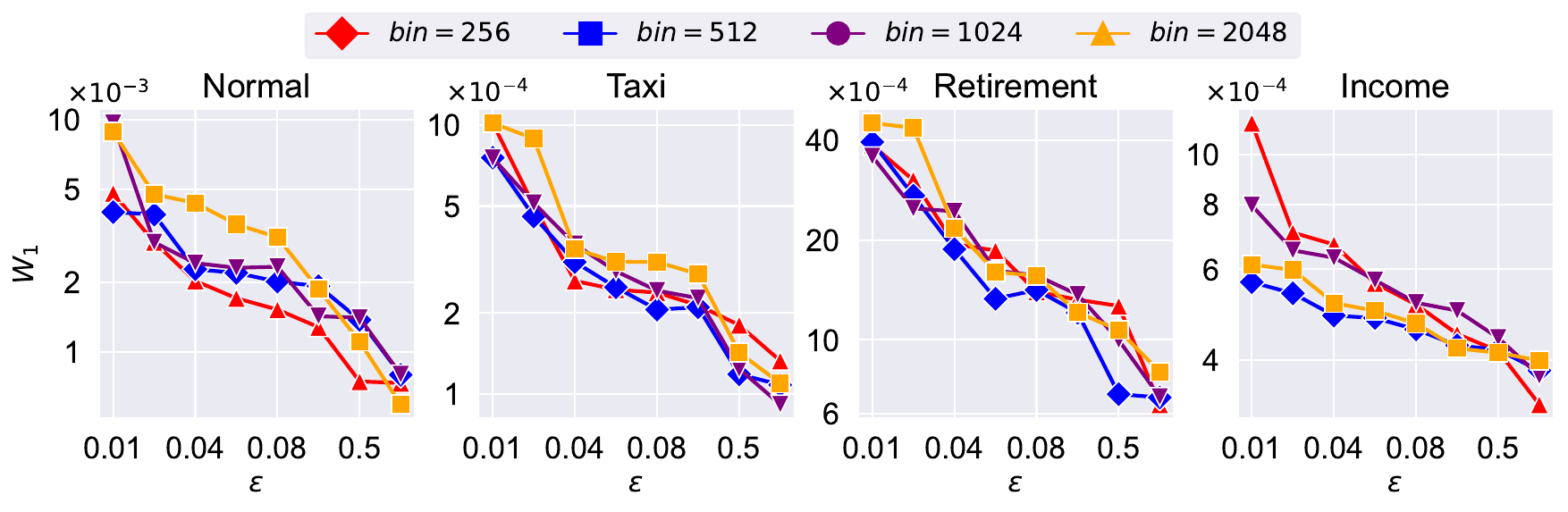}
    \vspace{-5pt}
    \caption{$W_1$ distance with different number of bins.}
    \label{fig:binsize_test}
    \vspace{-10pt}
\end{figure}

We empirically study the impact of $R$ and $\sigma_2$ in Figure~\ref{fig:hyperparameter_test} and the impact of different number of bins on four datasets in Figure~\ref{fig:binsize_test}.
We compare the default $\sigma_2$ against the varying $a$ and fixed $\sigma_2$.
Figure~\ref{fig:hyperparameter_test} shows that smaller or larger $R$ leads to higher error due to larger fluctuation in iterations or loss of distribution details.
Besides, different datasets have different optimal number of bins.

\noindent\textbf{Utility on Different Size of Dataset.}
In addition to $\epsilon$, privacy amplification is another factor influencing utility, which depends on the size $n$ of dataset.
Therefore, we also study the utility under varying size $n$ by sampling a subset from \textit{Income}. The result in Figure~\ref{fig:different_size_w1} demonstrates the ASP is always better than other protocols even under small datasets.

\subsection{Extension to High-Dimensionality and Privacy Models}
\noindent\textbf{High-dimensional Data Exploration.}
We extract eight other types of income information, excluding personal income, from \textit{Income} as a high-dimensional dataset. We compute the average $W_1$ distance of all 2-way and 3-way marginals.
As shown in Table~\ref{tab:high_dimensional}, ASP has small error under low dimensionality but the error increase as the dimensionality grows, which highlights the need for better protocols for high-dimensional data.

\begin{table}[!tbp]
\caption{$W_1$ distance on high-dimensional data.}
\label{tab:high_dimensional}
\begin{tabular}{lccccc}
\hline
\multicolumn{1}{c}{$\epsilon$} & 0.02   & 0.06    & 0.1      & 0.5     & 1       \\ \hline
2-way marginal                 & 0.0015 & 0.00079 & 0.000557 & 0.00053 & 0.00036 \\
3-way marginal                 & 0.012  & 0.00205 & 0.0013   & 0.00064 & 0.00043 \\ \hline
\end{tabular}
\end{table}

\noindent\textbf{Comparison with other privacy models and estimators.}
\begin{figure}[!tbp]
    \centering
    \includegraphics[scale = 0.3]{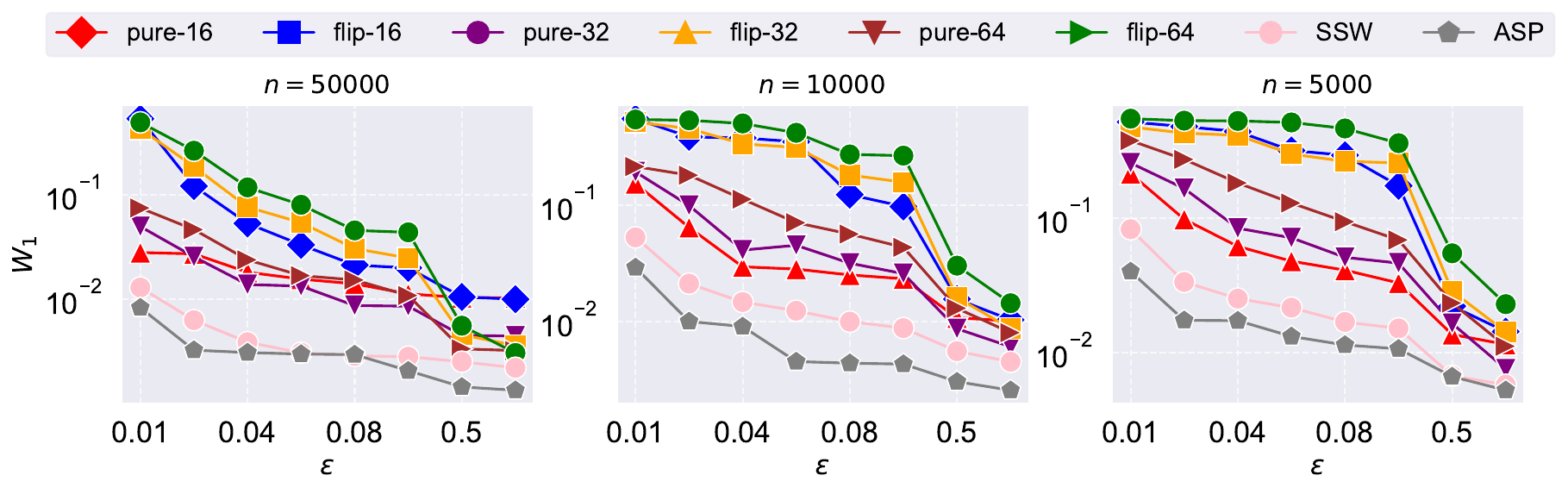}
    \caption{$W_1$ distance with different sample size of \textit{Income}.}
    \label{fig:different_size_w1}
    \vspace{-10pt}
\end{figure}
\begin{table}[!tbp]
	\caption{Comparison with other protocols}
	\label{tab:comparison_other_protocols}
	\begin{tabular}{lccccc}
		\hline
		\multicolumn{1}{c}{$\epsilon$}  & 0.02     & 0.06     & 0.1      & 0.5      & 1        \\ \hline
		SW (LDP)    & 3.34E-01 & 8.92E-02 & 4.73E-02 & 9.29E-03 & 3.84E-03 \\
		GM-64 (DP) & 4.34E-04 & 1.21E-04 & 7.03E-05 & 1.11E-05 & 5.31E-06 \\
		GM-32 (DP) & 2.18E-04 & 6.12E-05 & 3.79E-05 & 7.18E-06 & 3.32E-06 \\
		GM-16 (DP) & 1.16E-04 & 3.78E-05 & 2.26E-05 & 4.63E-06 & 2.56E-06 \\
		Trim (10\%) & 4.85E-02 & 5.49E-02 & 6.89E-02 & 1.43E-01 & 1.51E-01 \\
		Trim (20\%) & 1.34E-01 & 2.29E-01 & 2.89E-01 & 3.26E-01 & 3.26E-01 \\
		Trim (30\%) & 2.47E-01 & 3.89E-01 & 3.99E-01 & 3.99E-01 & 3.96E-01 \\
		SPM         & 1.33E-03 & 6.13E-04 & 5.33E-04 & 5.14E-04 & 4.43E-04 \\
		ASP (ours)  & 5.39E-04 & 4.83E-04 & 4.28E-04 & 4.19E-04 & 3.82E-04 \\ \hline
	\end{tabular}
	\vspace{-10pt}
\end{table}
We adapt PM protocol~\cite{wang2019collecting} to shuffle model (SPM) by privacy amplification bound~\cite{balle2019privacy}, and compare ASP with 1) SW in LDP, 2) Gaussian mechanism (GM) in central-DP with 16, 32, 64 domain chunks, 3) trim estimator discarding 5\%, 10\%, 15\% lowest and highest noisy samples, 4) SPM. As shown in Table~\ref{tab:comparison_other_protocols}, ASP outperforms all competitors except GM which only injects noise with constant scale.

%% file: robustness_evaluation_framework.tex
\section{Robustness Evaluation Framework}\label{sec:robustness_evaluation_framework}
In this section, we introduce the evaluation framework to assess the robustness of the shuffle DP protocols for numerical data under the data poisoning attack.
The framework inherits the idea of the attack-driven framework used in LDP~\cite{li2024robustness}, which compares the attack efficacy on different protocols and larger attack effectiveness indicates lower robustness.
In the rest of this section, We first introduce the threat model and propose a new robustness metric which can quantify attacks with different targets and enables us to learn and compare the resilience of shuffle DP protocols. Then, we elaborate on attacks for different shuffle DP protocols.

\subsection{Threat Model}

\noindent\textbf{Attacker's ability.}
To be consistent with the widely-used attacker assumption in the prior study~\cite{li2023fine, li2024robustness, cao2021data, cheu2021manipulation, cheu2022differentially}, we assume the shuffler $\mathcal{S}$ is an independent part and does not collude with the user or the server.
This is a reasonable assumption since there have been many secure shufflers~\cite{bogdanov2008sharemind, laur2011round, wang13improving} to mitigate such a collusion.
We also assume the attacker can control $n_f = n\beta (0 \leq \beta \leq 1)$ users from $n$ users to send fake values $\hat{Y} = [\hat{\bm{y}}^{(1)}, ..., \hat{\bm{y}}^{(n_f)}]$ to the server, where $\beta$ is the fraction of the fake user and $\hat{\bm{y}}^{(i)}=[\hat{y}^{(i)}_{1}, ...,\hat{y}^{(i)}_{w}]$ since each user can send $w$ messages.
Since the randomizer is deployed on the local side, the attacker can access the perturbation step and learn parameters, such as $\epsilon, \delta$, $m$, etc.


\noindent\textbf{Attacker's Objective.}
Different from the prior LDP robustness study~\cite{li2024robustness} that only focuses on the distribution-shift attack (DSA) where the attacker can only shift the distribution to the right or left end of the domain, we consider a more flexible and practical \textit{multimodal attack} which aims to manipulate the distribution such that it is concentrated on multiple attacker-desired targets in set $\bm{T} = \{T_1, T_2, ... \}$ as much as possible.
Our attack goal is more general and compatible with the existing DSA.
When the target $\bm{T} = \{1\}$ or $\{0\}$, the attacker can accomplish the same goal as DSA, shifting the distribution toward the right or left end of the domain.
Besides, our attack poses severe security threat to the real-world application. For example, the attacker can trick consumers into downloading malicious apps by shifting their rate distribution to the higher end; and an e-commerce company can manipulate the rate distribution of the commodity of its rival company, making it concentrate on a value smaller than its own commodity.
On the other hand, our attack is inspired by the prior attack~\cite{cao2021data} for frequency estimation and generalizable by focusing on maximizing the skewness on arbitrary targets over the distribution.

\subsection{Attack Effectiveness Metric}



A straightforward metric is to directly measure the statistical divergence between distributions before and after the attack.
However, the statistical divergence usually measures the high-level distance between two distributions, resulting in an inability to reflect how concentrated the distribution is at the target $\bm{T}$.
Put in another way, any attack that skews the distribution, whether making it concentrate at $\bm{T}$ or not, can increase the statistical divergence.
To address this problem, we propose a new metric called \textit{Real and Ideal Attack Ratio} (\RIAR) to quantify the normalized attack gain from the perspective of the target $\bm{T}$.
The main idea of \RIAR is to adopt an ideal attack that has the maximal attack efficacy and use it as the upper bound of the attack gain, and compare the upper bound with the gain of the attack launched by the attacker. To be consistent with the terminology of the ideal attack, we refer to the attacker-launched attack as the real attack.
The closer the real attack effectiveness is to the upper bound, the higher the attack performance and lower robustness of the protocol.

\begin{figure}
    \centering
    \includegraphics[scale = 0.29]{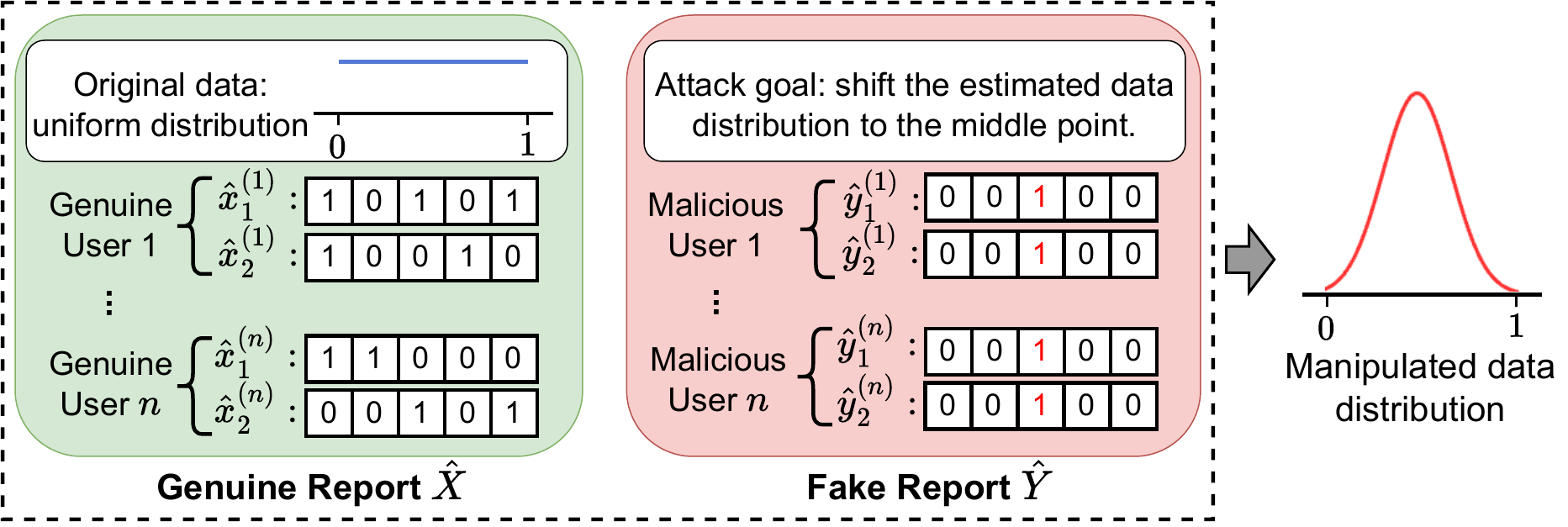}
    \caption{Illustration of the attack on Flip with $s=1$.}
    \vspace{-15pt}
    \label{fig:attack_illustration}
\end{figure}

Given the target set $\bm{T}$, the ideal attack can maximally shift the distribution to $\bm{T}$ and output an ideally-skewed distribution $\bm{f}_{ide}$, where all probability densities uniformly concentrate on each $T \in \bm{T}$. We note that, to incorporate different priorities of targets, $\bm{f}_{ide}$ could assign different densities to different $T \in \bm{T}$. However, for simplicity, we assume that every $T \in \bm{T}$ has the same priority.
Therefore, the upper bound of the attack gain is the distance $W_1(\bm{f}, \bm{f}_{ide})$ between the original distribution $\bm{f}$ and the ideally-skewed distribution $\bm{f}_{ide}$. We consider the Wasserstein distance (denoted as $W_1(\cdot)$) since it measures the overall difference between two distributions. After the attack, we also measure the distance $W_1(\hat{\bm{f}}_a, \bm{f}_{ide})$ between the distribution $\hat{\bm{f}}_a$ skewed by the real attack with the $\bm{f}_{ide}$ to quantify the absolute gap between the real attack gain and the upper bound. By calculating the ratio of $W_1(\hat{\bm{f}}_a, \bm{f}_{ide})$ to $W_1(\bm{f}, \bm{f}_{ide})$, we normalize this gap and have the formal definition $\RIAR = \frac{W_1(\hat{\bm{f}}_a, \bm{f}_{ide})}{W_1(\bm{f}, \bm{f}_{ide})}$.
Lower \RIAR means better attack effectiveness and worse robustness. This is because the denominator in \RIAR is independent of the real attack, and a stronger attack only makes the numerator smaller and results in lower \RIAR.
In particular, \RIAR equal to zero indicates that the real attack is equivalent to the ideal attack.
Moreover, involving the ideal attack efficacy $\bm{f}_{ide}$ brings a benefit in robustness analysis. \RIAR defines the upper bound of the attack given the target $\bm{T}$, enabling us to better understand the protocol robustness under various attacks.

\begin{figure*}[!htbp]
    \centering
    \includegraphics[scale = 0.3]{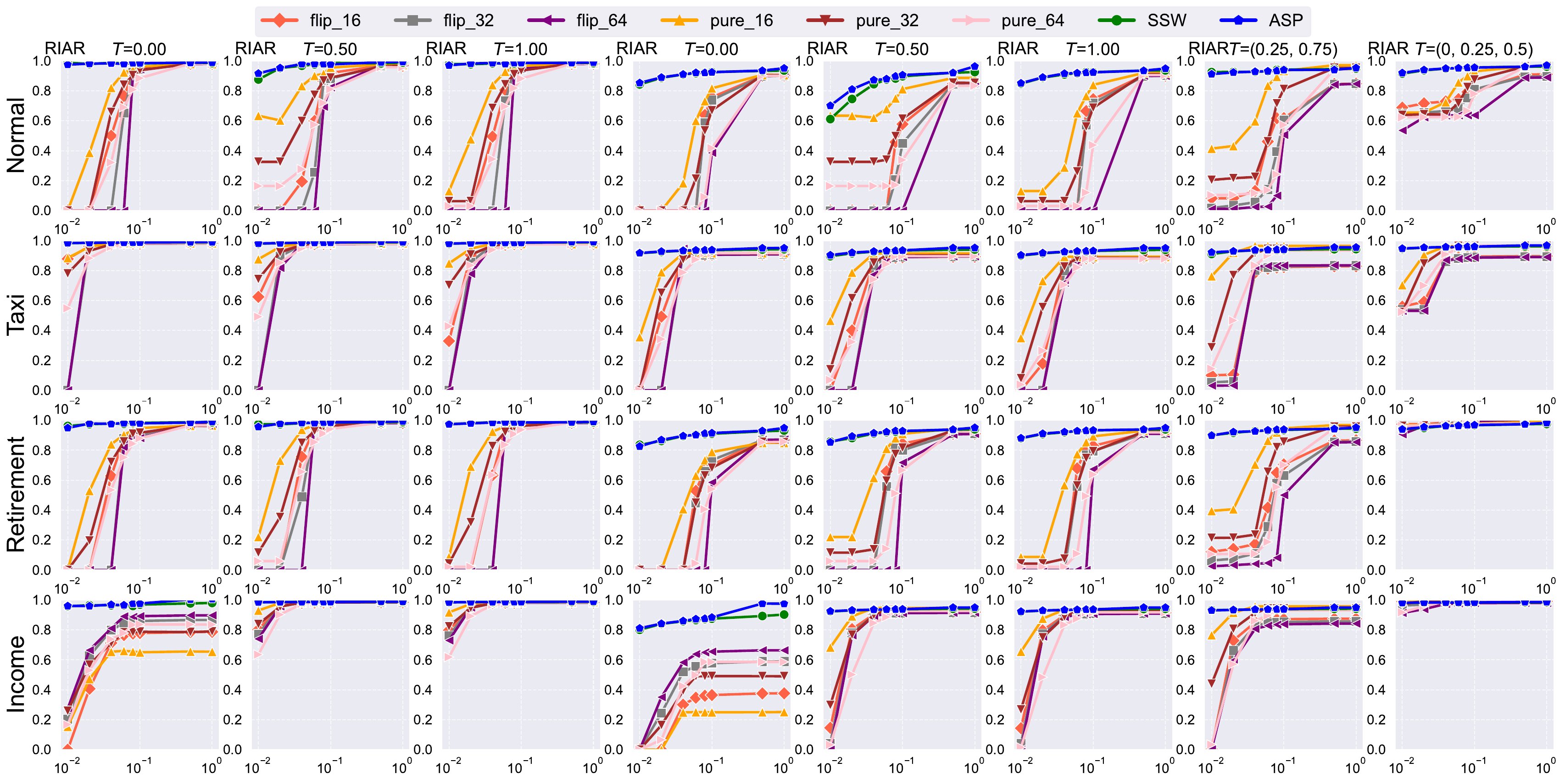}
    \vspace{-5pt}
    \caption{Robustness results with varying $\epsilon$. Each row corresponds to one dataset.}
    \vspace{-5pt}
    \label{fig:robustness_with_eps}
\end{figure*}

\subsection{Attack Details}
We first introduce the attack to SCFO with binning and consistency, and then describe the attack to distribution reconstruction-based protocols.


It is challenging to find the optimal attack for SCFOs to shift all densities to the target $\bm{T}$, because it depends on the actual data distribution which is unknown to the attacker.
Nevertheless, the main idea of the attack is to increase the frequencies of every $T \in \bm{T}$.
Here we use set $\bm{I}_{\bm{T}}$ to denote indices of all target bins containing $T \in \bm{T}$.
Therefore, we consider a heuristic attack that increases the frequencies of all targeted values $\mathcal{A}(B_{\bm{I}_{\bm{T}}}) = \sum_{i \in \bm{I}_{\bm{T}}} \mathcal{A}(B_{i})$.
This attack guarantees that the distribution can be shifted to $\bm{T}$ and empirically achieves the significant attack efficacy.
We describe the $\mathcal{A}(B_{\bm{I}_{\bm{T}}})$ given $n_f$ compromised users and the fake value creation, and illustrate the attack on Flip in Figure~\ref{fig:attack_illustration}.

\vspace{3pt}

\textbf{Attack on Flip.}
The aggregated result $\mathcal{A}(B_{\bm{I}_T})$ and the fake value creation are as follows.
    

\begin{itemize}[leftmargin = *]
    \item \textit{Aggregation $\mathcal{A}(B_{\bm{I}_T})$}:
    $\mathcal{A}(B_{\bm{I}_T})$ is the summation of all frequencies of target bins. Formally, $\mathcal{A}(B_{\bm{I}_{\bm{T}}}) = \sum_{i_T \in \bm{I}_{\bm{T}}} ( \frac{1}{n(1-2q)} [\sum_{i=1}^{n-n_f}\sum_{v=1}^{k+1} \hat{\bm{x}}^{(i)}_{v}[i_T] + \sum_{i=1}^{n_f}\sum_{v=1}^{k+1} \hat{\bm{y}}^{(i)}_{v}[i_T] - n(k+1)q ] )$.
    
    \item \textit{Fake value creation}: To maximize $\mathcal{A}(B_{\bm{I}_{\bm{T}}})$, the most effective way is that the $i$-th fake user initializes $k+1$ vectors $[\hat{\bm{y}}_v^{(i)}]_{v=1}^{k+1}$ of all 0's, and for each vector, sets $\hat{\bm{y}}^{(i)}_{v}[j] = 1$ for $\forall j \in \bm{I}_{\bm{T}}$. Then, the summation $\sum_{i=1}^{n_f}\sum_{v=1}^{k+1} \hat{\bm{y}}^{(i)}_{v}[j] = n_f (k+1)$ for every target bin is maximal.
\end{itemize}




\textbf{Attack on Pure.}
The aggregated result $\mathcal{A}(B_{\bm{I}_{\bm{T}}})$ and the fake value creation are as follows.

\begin{itemize}[leftmargin = *]
    \item \textit{Aggregation $\mathcal{A}(B_{\bm{I}_{\bm{T}}})$}: Pure uses $s$ dummy points to achieve shuffle DP, thereby the aggregation is $\mathcal{A}(B_{\bm{I}_{\bm{T}}}) = \sum_{i_T \in \bm{I}_{\bm{T}}} [(\sum_{i=1}^{n - n_f} \sum_{k=1}^{s+1} \mathbb{I}_{[i_T]}(\hat{\bm{x}}^{(i)}_k) + \sum_{i=1}^{n_f} \sum_{k=1}^{s+1} \mathbb{I}_{[i_T]}(\hat{\bm{y}}^{(i)}_k) - \frac{ns}{m})]$.

    \item \textit{Fake value creation}: The most effective way to promote $\mathcal{A}(B_{\bm{I}_{\bm{T}}})$ is that each fake user picks any value $i_T$ in $\bm{I}_{\bm{T}}$ as $\hat{\bm{y}}^{(i)}_k$ and maximize the term $\sum_{i_T \in \bm{I}_{\bm{T}}} \sum_{i=1}^{n_f} \sum_{k=1}^{s+1} \mathbb{I}_{[i_T]}(\hat{\bm{y}}^{(i)}_k) = n_f (s + 1)$.
\end{itemize}

\begin{figure*}[!htbp]
    \centering
    \includegraphics[scale = 0.3]{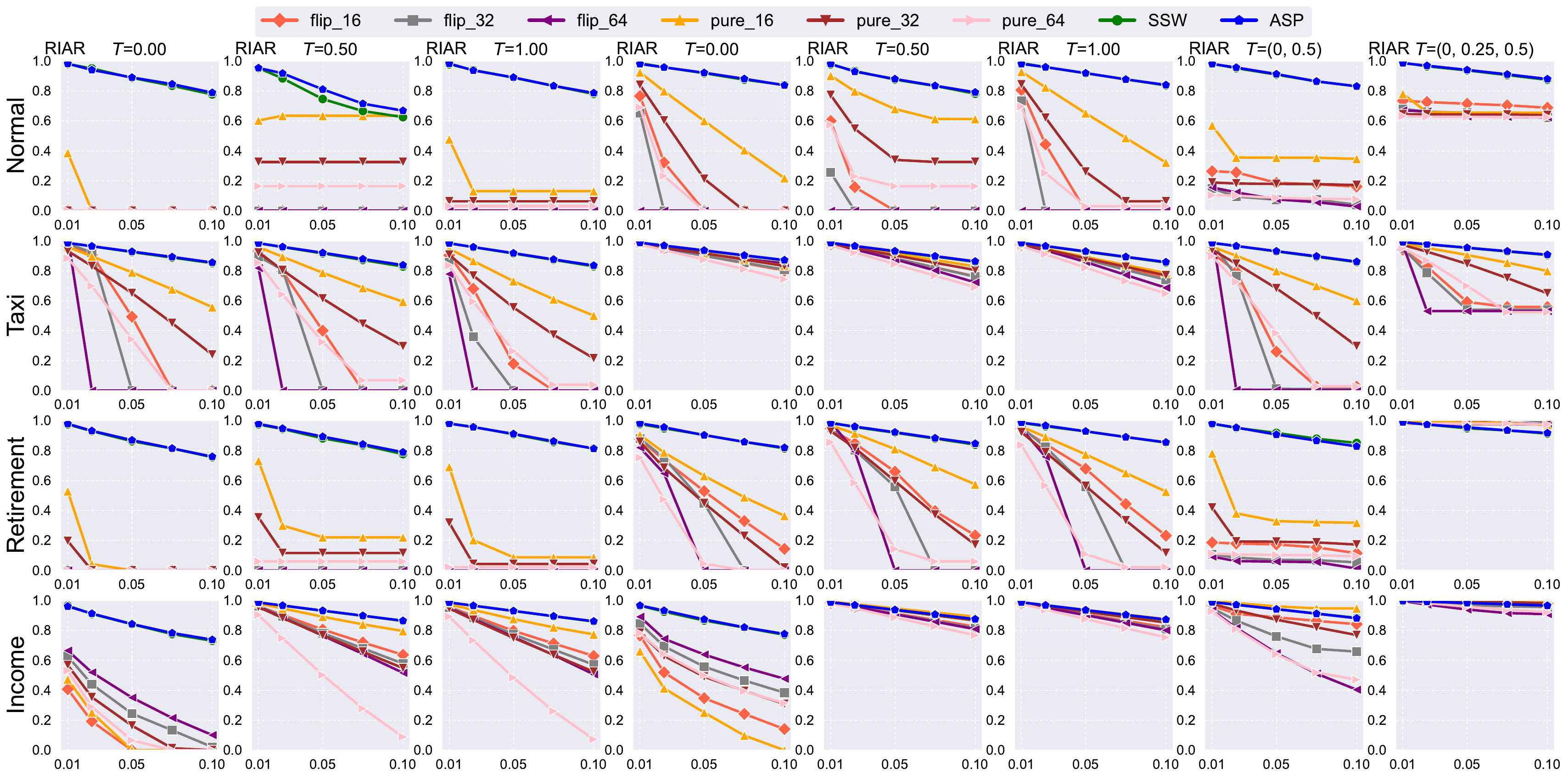}
    \vspace{-5pt}
    \caption{Robustness results with varying $\beta$. Each row corresponds to one dataset.}
    \vspace{-15pt}
    \label{fig:robustness_with_beta}
\end{figure*}


\textbf{Attack on SSW and ASP.}
The variant of EM algorithm has no closed-form expression and it depends on $\epsilon$ and true data distribution which is unknown to the attacker.
Therefore, it is non-trivial for the attacker to predict the estimate for SSW and ASP, and design the optimal attack strategy for them.
Similar to the attack on SCFOs, the crux to shift the probability density to $\bm{T}$ is still to promote the frequencies of values in $\bm{T}$ while simultaneously reducing densities over other values accordingly.
To this end, we consider an intuitive strategy to launch the attack.
We inject fake values around each $T \in \bm{T}$ over the output domain of SSW and ASP.
Specifically, we uniformly at random inject fake values into three ranges: 1) $[T - b, T + b]$, 2) $[T - \frac{b}{2}, T + \frac{b}{2}]$ and $[T - \frac{b}{3}, T + \frac{b}{3}]$.

\noindent\textbf{Robustness discussion.}
We analyze the robustness property of our EMAS. By comparing with the traditional EM algorithm, Theorem~\ref{the:robustness_analysis} presents that our EMAS can reduce the attack efficacy when fake values significantly increase the frequencies of polluted bins.
We also use the notation $\hat{f}_i$ to denote the estimated result after EM algorithm but before the AS-step.

\begin{theorem}\label{the:robustness_analysis}
    If in the $t$-th iteration, the polluted bin $B_i$ has the maximum frequency $\hat{f}_i^t = \max_{j \in [i-R, i+R]} \hat{f}_j^t$ among bins within radius $R$, the output $\tilde{f}_i^t$ of EMAS will be smaller than $\hat{f}_i^t$ and reduce the attack impact.
\end{theorem}

\begin{proof}
    Given the process of EMAS, the EMAS result is $\tilde{f}_i^t = \sum_{j \in [i-R, i+R]} w_j \hat{f}_j^t$.
    Since $\hat{f}_i^t = \max_{j \in [i-R, i+R]} \hat{f}_j^t$, we have
    \begin{align*}
        \hat{f}_i^t = \sum_{j \in [i-R, i+R]} w_j \hat{f}_i^t \geq \sum_{j \in [i-R, i+R]} w_j \hat{f}_j^t = \tilde{f}_i^t.
    \end{align*}
\end{proof}

\begin{figure*}[!htbp]
    \centering
    \includegraphics[scale = 0.3]{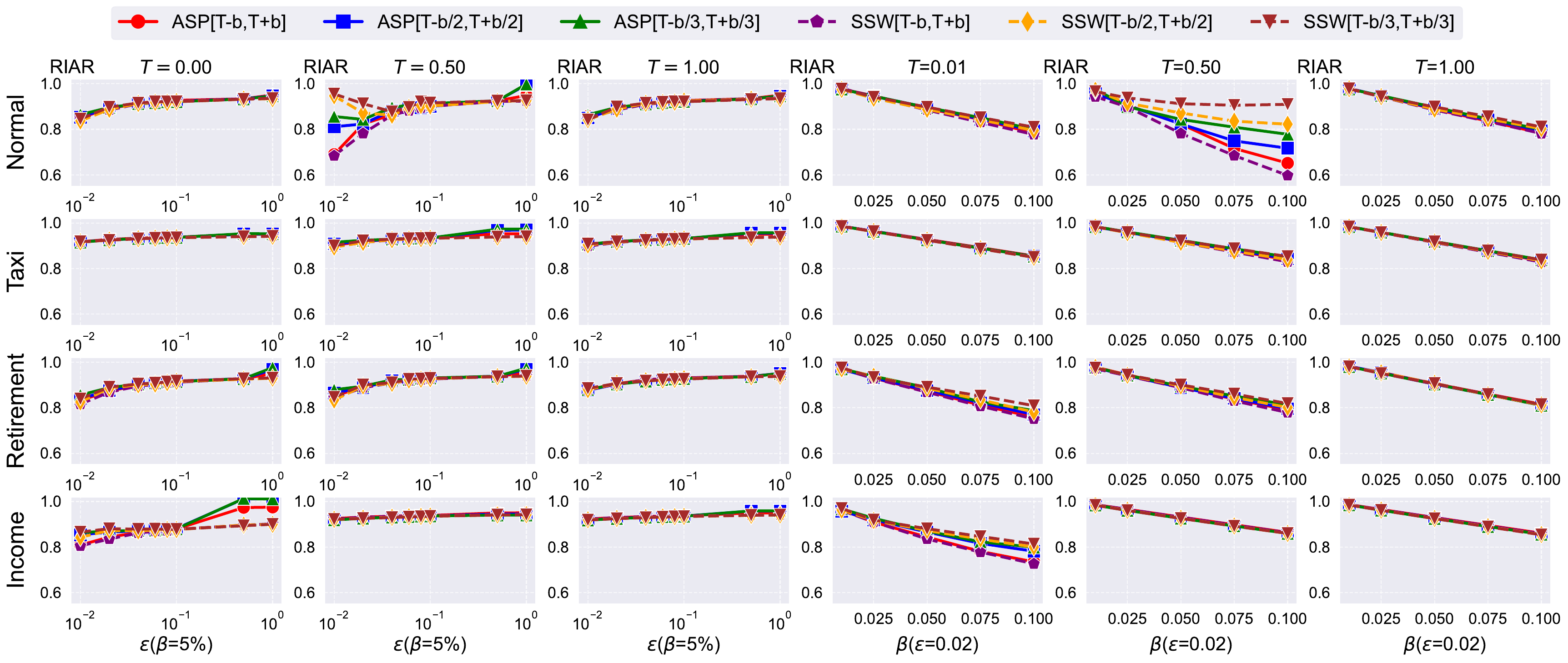}
    \vspace{-5pt}
    \caption{Different attack strategies to SSW and ASP.}
    \vspace{-15pt}
    \label{fig:asp_ssw_various_attack}
\end{figure*}

%% file: robustness_evaluation_v2.tex
\vspace{-7pt}
\section{Robustness Evaluation}\label{sec:robustness_evaluation}

\subsection{Experiment Setup}
\noindent\textbf{Datasets.}
We use the same datasets as utility evaluation.

\noindent\textbf{Parameter setting.}
We vary $\epsilon$ from 0.01 to 1 to cover a wide range of privacy level, and set $\beta$ from 1\% to 10\% to study the impact of different attack intensities.
We also set $\bm{T} = \{0\}, \{0.5\}, \{1\}$ for attacks on left-end, middle and right-end over the domain, and set $\bm{T} = \{ 0, 0.5 \}$ and $\bm{T} = \{ 0, 0.25, 0.5 \}$ for multimodal attacks.
We repeat the experiment 100 times and show the average result.

\vspace{-7pt}
\subsection{Results}

The robustness comparison is shown in Figure~\ref{fig:robustness_with_eps} and \ref{fig:robustness_with_beta}.
We attack the SSW and ASP by selecting three different ranges and choosing the range $[T-b, T+b]$ that yields the lowest \RIAR for a fair comparison.
We defer the discussion on the attack performance on SSW and ASP with different ranges to Figure~\ref{fig:asp_ssw_various_attack}.
We only report result on datasets \textit{Normal} and \textit{Income} here due to the page limits and the rest is shown in the full version~\cite{full_version}.
We have following key observations.
\begin{itemize}[leftmargin=*]
    \item Robustness comparison in Figure~\ref{fig:robustness_with_eps} and \ref{fig:robustness_with_beta}.
    \begin{itemize}
    \item As shown in Figure~\ref{fig:robustness_with_eps} and \ref{fig:robustness_with_beta}, ASP and SSW exhibit the highest \RIAR, and ASP has a slightly higher RIAR than SSW, because ASP adopts dynamic weights to more effectively average the frequencies of the polluted bins.

    \item For multimodal attacks, our ASP still has the best robustness. Besides, the RIAR for other protocols increases as the number of targets grows. This is because ideally promoting multiple targets is harder than single target.
    
    \item In Figure~\ref{fig:robustness_with_eps}, the robustness is getting worse (smaller \RIAR) when $\epsilon$ is smaller. This is because more bogus data in Flip/Pure is sent under smaller $\epsilon$; for SSW/ASP, smaller $\epsilon$ indicates that more values may be restored to the target and increase attack efficacy.

    \item In Figure~\ref{fig:robustness_with_beta}, the robustness of the protocol decreases as $\beta$ increases. This is because a larger $\beta$ enables more attackers to better manipulate the final estimate.


    \item In most cases, SCFO-based protocols with larger number of bins $c$ have worse robustness, as attackers can precisely control the poisoning data given finer granularity.


    \item On \textit{Normal} dataset, the attack efficacy is similar at $T = 0$ and 1 but different from $T = 0.5$, because normal distribution is symmetric with 0.5.
    \end{itemize}
    

    \item In Figure~\ref{fig:asp_ssw_various_attack}, injecting fake values into $[T-b, T+b]$ exhibits the strongest efficacy, since it can match the perturbation probability to effectively promote targets.
    
\end{itemize}

%% file: related_work.tex
\vspace{-7pt}
\section{Related Work}
\noindent\textbf{Shuffler-DP protocols for categorical and numerical data.}
Frequency estimation is a basic and popular task for categorical data in shuffle-DP.
There have been several protocols~\cite{cheu2022differentially, li2023privacy, wang13improving, balcer2021connecting, balle2019privacy} for this task under the pure shuffle model, in which the shuffler only permutes the user report.
Among them, Flip~\cite{cheu2022differentially} and Pure~\cite{li2023privacy} are two state-of-the-art protocols.
Besides the pure shuffle model, Takao \textit{et al.}~\cite{murakami2025augmented} proposed the augment shuffle model to improve the utility at the expense of stronger security requirements.
Moreover, there are also protocols~\cite{ghazi2021differentially, balle2020private, dong2024almost} that can handle numerical data, but focusing on specific summation task.
In this paper, we propose a new protocol under pure shuffle model for accurate and robust distribution estimation on numerical data.

\noindent\textbf{Post-processing methods.}
Post-processing methods allow one to utilize the structural information to improve the utility without extra privacy cost.
Prior works~\cite{kulkarni2019answering, wang2018privtrie} proposed hierarchy structure to post process the estimated histogram for accurate range query.
Jia \textit{et al.}~\cite{jia2019calibrate} propose to use a prior knowledge about the dataset.
Wang \textit{et al.}~\cite{wang2020locally} develop a suite of methods to constrain the frequencies to be non-negative and summing-up-to-one.
Li \textit{et al.}~\cite{li2020estimating} and Ye \textit{et al.}~\cite{yerevisiting} propose two variants of EM algorithms to find the maximum likelihood estimate by the noisy report.
We also develop a new post-processing EMAS that can guarantee utility and robustness. 

\noindent\textbf{Robustness Analysis on LDP protocols.}
There is a line of work on robustness analysis about LDP protocols.
Cheu \textit{et al.}~\cite{cheu2021manipulation} studied the issue of data manipulation and showed inherent vulnerability in non-interactive LDP protocols. The attacks on frequency estimation and key-value data collection were also studied in~\cite{cao2021data} and~\cite{wu2022poisoning}. They focused on maximizing the statistics of the attacker-chosen items. In contrast to straightforward maximization, Li \textit{et al.}~\cite{li2023fine} considered an attack that allows the attacker to fine-tune the final estimate to a target value for the mean and variance estimation. Li \textit{et al.}~\cite{li2024robustness} also analyze the robustness of LDP protocols for distribution estimation.
In this paper, we expand the existing robustness evaluation framework to study the robustness of shuffle-DP protocols for distribution estimation.

\vspace{-3pt}
\section{Conclusion}
We propose a novel shuffle-DP protocol called ASP for distribution estimation on numerical data.
In ASP, we leverage the characteristic of the shuffler to further optimize the parameter for randomizer, and introduce a new aggregation EMAS based on the adaptive smoothing, which achieves consistently preferable performance across different datasets.
To evaluate the protocol robustness, we develop a new framework which uses the more general attack and target-dependent robustness metric to provide a holistic robustness assessment.
Extensive experimental evaluations demonstrate that ASP simultaneously performs the best under three dimensions of utility, message complexity and robustness.

\newpage

\section{AI-Generated Content Acknowledgement}
We acknowledge that we do not use the AI-generated content in the manuscript, and all scientific contributions are conducted by the authors.